\pdfoutput=1
\documentclass[acmsmall,screen,nonacm,appendix]{acmart}

\usepackage{iris}
\usepackage{setup}
\usepackage{pftools}
\usepackage{mathpartir}

\setcopyright{rightsretained}
\acmDOI{10.1145/3632854}
\acmYear{2024}
\copyrightyear{2024}
\acmSubmissionID{popl24main-p66-p}
\acmJournal{PACMPL}
\acmVolume{8}
\acmNumber{POPL}
\acmArticle{12}
\acmMonth{1}

\begin{document}

\author{Dan Frumin}
\email{d.frumin@rug.nl}
\orcid{0000-0001-5864-7278}
\affiliation{\institution{University of Groningen}
  \country{The Netherlands}
}

\author{Amin Timany}
\orcid{0000-0002-2237-851X}             \affiliation{            \institution{Aarhus University}
\country{Denmark}                    }
\email{timany@cs.au.dk}          

\author{Lars Birkedal}
\orcid{0000-0003-1320-0098}             \affiliation{
\institution{Aarhus University}            \country{Denmark}                    }
\email{birkedal@cs.au.dk}

\title{Modular Denotational Semantics for Effects with Guarded Interaction Trees}

\begin{abstract}
  We present \emph{guarded interaction trees} --- a structure and a fully formalized framework for representing higher-order computations with higher-order effects in Coq, inspired by domain theory and the recently proposed interaction trees.
  We also present an accompanying separation logic for reasoning about guarded interaction trees.
  To demonstrate that guarded interaction trees provide a convenient domain for interpreting higher-order languages with effects, we define an interpretation of a PCF-like language with effects and show that this interpretation is sound and computationally adequate; we prove the latter using a logical relation defined using the separation logic.
  Guarded interaction trees also allow us to combine different effects and reason about them modularly.
  To illustrate this point, we give a modular proof of type soundness of cross-language interactions for safe interoperability of different higher-order languages with different effects.
  All results in the paper are formalized in Coq using the Iris logic over guarded type theory.
\end{abstract}




\maketitle

\section{Introduction}
\label{sec:introduction}
Interaction trees \cite{XiaEtAl:2019} are a recently proposed formalism for representing and reasoning about (possibly) non-terminating programs with side effects in Coq (a terminating type theory without effects).
Since its inception, interaction trees have been applied, including but not limited, to specifying and verifying network servers \cite{KohEtAl:2019,ZhangEtAl:2021}, semantics of LLVM \cite{ZakowskiEtAl:2021}, semantics of a language for robotics \cite{YeEtAl:2022}, non-interference \cite{SilverEtAl:2023}, and verification of concurrent objects with transactional memory \cite{LesaniEtAl:2022}.

The introduction of interaction trees was motivated by a desire to simplify mechanized formalizations of interactive, effectful, non-terminating computations and the developers of the ITrees library argued that ITrees can represent computations in a way which is more \emph{modular} than representations based on operational semantics and \emph{executable} (in contrast to earlier representations based on traces represented as predicates on events).
In particular, the idea is that interaction trees can be used to give \emph{denotational semantics} to programming languages and thus allow one to abstract away from syntactic details and reuse meta-language features such as function composition so as to obtain more robust mechanizations.
And, indeed, the applications mentioned above demonstrate that interaction trees work well for giving semantics to first-order programming languages with first-order effects.

The challenge we address in this paper is that interaction trees cannot easily be used as a model of \emph{higher-order} programming languages with \emph{higher-order effects}, which, of course, limits the applicability of interaction trees.
Indeed, the ease of use of interaction trees is enabled, in part, by two restrictions imposed on the computations represented by the interaction trees: the computations must be first-order, and the effects that the computation performs must be first-order as well.
With those restrictions, the type of interaction trees forms a monad, which allows one to compose the represented computations and reason about them modularly.
(In principle, one could represent higher-order computations by means of closures in interaction trees, but that would defeat the purpose of interaction trees and force one to reason about syntactic representations, which interaction trees otherwise relieves one from.)
To understand the limitations to first-order programs and first-order effects, we call to mind the definition of interaction trees.

Interaction trees are possibly infinite trees with two types of branching.
The first type of branching represents a ``delayed'' computation (similar to that of the delay monad), or a computation performing a silent step.
The second type of branching represents a computation that performs an effect; different results of the effect lead to different branches.
Interaction trees are formalized as coinductive types in Coq, allowing one to leverage existing infrastructure for coinductive programs and proofs:
\vspace{-0.6em}
\begin{lstlisting}[language=Coq, basicstyle=\small]
  CoInductive itree (E : Type -> Type) (R : Type) :=
  | Ret : R -> itree E R
  | Tau : itree E R -> itree E R
  | Vis {A : Type} : E A -> (A -> itree E R) -> itree E R
\end{lstlisting}
\vspace{-0.8em}
Now the point is that if we wish to represent higher-order computations, then we cannot simply add a constructor \coqe{Fun : (itree E R -> itree E R) -> itree E R}, as the resulting recursive type would have negative occurrences of the recursive variable (the \coqe{itree E R} on the left of the first arrow).
Similarly, if we want to support computations with higher-order effects, i.e.,
the result of an effect is an interaction tree itself, we run into the same problems with positivity.
For example, in the following potential signature for a higher-order effect, the parameter \coqe{test} occurs in a negative position:
\vspace{-0.6em}
\begin{lstlisting}[language=Coq,basicstyle=\small]
  Inductive test : Type -> Type :=
  | T : nat -> test (itree test unit).
\end{lstlisting}
\vspace{-1.1em}

\paragraph{Guarded interaction trees: Iris and guarded type theory}
Our \emph{goal} is to address the challenge of extending interaction trees to allow for higher-order computations and higher-order effects, in such a way that we retain some of the advantages of interaction trees; in particular we wish to obtain a representation with which we can work efficiently in Coq.
From the discussion above, it is clear that a way forward is to work in a setting that allows to solve mixed-variance recursive domain equations. There are several possible choices for such a setting, including classical Scott domain theory \cite{Scott:1976,SmythPlotkin:1982} and guarded type theory \cite{BirkedalEtAl:2010a,DBLP:journals/corr/abs-1208-3596}.
We choose to use the latter since this choice allows us to leverage the Iris program logic framework in Coq and thence obtain
an efficient environment in which we can work efficiently and formally in Coq.

Thus in this paper we introduce \emph{guarded interaction trees}, which are formally defined in guarded type theory as a solution to a guarded recursive domain equation, and we show how guarded interaction trees can be used to represent higher-order computations and higher-order effects.
Moreover, we demonstrate how we can retain some of the benefits of interaction trees, in particular modularity with respect to effects and ease of use in Coq.
The extension to higher-order computations and effects does come with a certain price, in that we need to reason about guardedness, but we believe the use of Iris alleviates this.

Our Coq formalization is available online at
\begin{center}
  \url{https://github.com/logsem/gitrees/tree/popl24}.
\end{center}
(tag \texttt{popl24} in the Git repository)

\paragraph{Contributions.}
In this paper we present the following contributions, all formalized as part of our extensible and adaptable Coq formalization:
\begin{enumerate}
\item We present \emph{guarded interaction trees}, describe the associated recursion principle, and demonstrate how to write combinators to program with guarded interaction trees (\Cref{sec:gitrees}).
\item We describe a way of \emph{reifying} effects in the guarded interaction trees, and the reduction semantics (\Cref{sec:reductions}).
\item We show how to give a model of a higher-order programming language with general recursion and effects in guarded interaction trees, and show that the model is sound (\Cref{sec:iolang_model}).
\item We build a separation logic (a program logic) on top of guarded interaction trees, allowing us to reason about their behavior (\Cref{sec:program_logic}).
\item We use the separation logic to show that the model that we construct in \Cref{sec:iolang_model} satisfies computational adequacy (\Cref{sec:adequacy}).
\item We demonstrate how multiple different effects can be combined in guarded interaction trees, and how the separation logic is used to reason about the effects locally (\Cref{sec:subeffects}).
\item Finally, we utilize the results above, and use guarded interaction trees to show type safety of cross-language interactions for safe interoperability of languages with different effects (\Cref{sec:interop}).
\end{enumerate}
We discuss related work in \Cref{sec:related_work}.
Before we present our results, we briefly go over some preliminaries about the setting that we are working in.

\section{Iris Logic over Guarded Type Theory}
\label{sec:guarded_logic}
In this section we describe the Iris logic, in which we shall define and work with guarded interaction trees.
Our treatment is brief since Iris has been described in many other papers and we are just using a small extension of
the usual presentation; we refer the reader to the literature on Iris \cite{DBLP:journals/jfp/JungKJBBD18}
and guarded type theory \cite{DBLP:journals/corr/abs-1208-3596} for more details. 

Iris is usually presented as a separation logic over a simple type theory.The model of Iris, however, models a richer type theory and in this paper we are going make use of that and consider Iris over a guarded type theory with
\begin{enumerate*}
\item a modicum of dependent type theory, and 
\item the ability to define guarded recursive types.
\end{enumerate*}
Both of these features are supported by the existing Coq implementation of Iris and the associated Iris proof mode \cite{DBLP:conf/popl/KrebbersTB17}.

Note that since we are working formally in Iris in Coq, there are two logical levels at play: the statements and proofs at the Coq level (which we refer to as the meta-logic level, or as the meta level), and the statements and proofs at the level of Iris (which we refer to as the logic level).

We recall the grammar of Iris in \Cref{fig:grammar}; the syntax consists of types, terms, and propositions.
Most of the grammar is standard for higher-order intuitionistic logic, with the parts related to guarded recursion highlighted in {\color{blue}blue}.
As usual in higher-order logic, we have a well-typedness judgment
$x_1 : \type_1, \dots, x_n : \type_n \proves \wtt{\term}{\type}$ stating that the term $\term$ has type $\type$, under the assumption that the variables $x_i$ have types $\type_i$.
In the grammar for types, ${\texttt I}$ ranges over so-called discrete types, which are meta-level types embedded into the types of Iris.
Note that types include dependent types over discrete types.  While we have not shown it in the grammar, we can also form types as solutions to guarded recursive domain equations, i.e., type equations where the recursive occurence of the type being defined is guarded under the $\latert$ type modality. Such types are defined up to isomorphism; we will see an example shortly: the type of guarded interaction trees will be such a recursive type and will be introduced in the following section.
A useful semantic intuition for the types of Iris is that they denote (certain kinds of) time-indexed sets, i.e.,
families of sets indexed over natural numbers. At time step $n>0$, the later type $\latert \type$ consists of the elements
of $\tau$ at a later time step, i.e., at $n-1$. At time step $n=0$, the type $\latert \type$ is a singleton set.
Intuitively, guarded recursive types exist because to understand what a guarded recusive type is at time step $n$,
one only needs to understand what it is at $n-1$, since the recursion is guarded. 

\begin{figure}[t]\small
  \begin{align*}
    \type \bnfdef{}&
     \Prop \mid
       0 \mid
      \Tunit \mid
      \Tbool \mid \Tnat \mid  
      \type + \type \mid
      \type \times \type \mid
      \type \to \type \mid {\color{blue}\latert \type} \mid
      I \mid               
      \Sigma_{\idx\in \mathtt{I}} \type_\idx \mid
      \Pi_{\idx\in \mathtt{I}} \type_\idx \mid
      \dots
    \\[0.4em]
    \term \bnfdef{}&
         \var \in \Var\mid
         \sigfn(\term_1, \dots, \term_n) \mid
         \textlog{abort}\; \term \mid
         () \mid
         (\term, \term) \mid
         \pi_i\; \term \mid
         \Lam \var:\type.\term \mid
         \term(\term)  \mid
    \\&
    \textlog{inj}_i\; \term \mid
    \textlog{match}\; \term \;\textlog{with}\; \Ret\textlog{inj}_1\; \var. \term \mid \Ret\textlog{inj}_2\; \var. \term \;\textlog{end} \mid
    {\color{blue}\Next(\term)} \mid
    {\color{blue}{\fix_\type}}
    \\[0.4em]
    \prop \bnfdef{}&
    \FALSE \mid
    \TRUE \mid
    \term =_\type \term \mid
\prop \land \prop \mid
    \prop \lor \prop \mid
    \prop \to \prop \mid
\Exists \var:\type. \prop \mid
    \All \var:\type. \prop \mid
    {\color{blue}\later \prop} \mid
    {\color{blue}\MU \var:\type. \prop}
  \end{align*}
  \caption{Grammar for the base logic.}
  \label{fig:grammar}
\end{figure}

Elements of $\latert \type$ can be constructed from elements of $\type$, using the $\Next$ constructor, and
we can form fixed points for guarded endo-functions:
\begin{mathpar}
  \infer{\vctx \proves \wtt{\term}{\type}}
  {\vctx \proves \wtt{\Next(\term)}{\latert \type}}
  \and
  \infer{}
  {\vctx \proves \wtt{\fix_\type}{(\latert \type \to \type) \to \type}}
\end{mathpar}
The $\latert$ type former is functorial and we write
$\wtt{\latert f}{\latert \type_1 \to \latert \type_2}$ for its action on terms $\wtt{f} {\type_1 \to \type_2}$.
The fixed point satisfies the equation $\fix_\type(f) = f(\Next(\fix_\type(f)))$.

For propositions $\wtt{\prop}{\Prop}$ we also have the provability judgment
$\vctx \mid \prop \proves \propB$ stating that $\propB$ is derivable from $\prop$ in the typing context $\vctx = x_1 :\type_1, \dots, x_n:\type_n$.
The rules corresponding to the intuitionistic fragment are standard.
Here we present the rules concerning the guarded part of the logic.

On the level of propositions, we have a(nother) \emph{later modality} $\later$.
This is the later modality most users of Iris are already familiar with.
It is related to the later modality on types in that $\later(t =_\type t') \provesIff \Next(t) =_{\latert\type} \Next(t')$.
We recall that $\later$ can be used to define guarded recursive predicates and and that it supports reasoning via L{\"o}b induction:
\begin{mathpar}
  \infer{\vctx\proves \wtt{P}{\Prop}}
  {\vctx\proves \wtt{\later P}{\Prop}}
  \and
  \infer{\vctx, x:\type \proves \wtt{P}{\Prop} \and \mbox{$x$ is guarded in $P$}}
  {\vctx \proves \wtt{\MU x:\type. P}{\Prop}}
  \and
  \infer
  {\vctx \mid \later P \proves P}
  {\vctx \mid \TRUE \proves P}
\end{mathpar}
Iris also includes separation logic connectives; we recall those later, when we need them, in \Cref{sec:program_logic}.

If $P$ is a proposition that consists only of intuitionistic logic connectives without $\later$, then we can interpret it both as a meta-level proposition (i.e. a Coq proposition), and as an Iris proposition.
For such propositions we have the following result, connecting Iris with the meta-level:
\begin{theorem}[Iris Adequacy]
  \label{thm:iris_adequacy}
Let $P$ be a proposition containing only intuitionistic connectives.
Then, if $\TRUE \proves P$ is derivable in Iris, then $P$ also holds at the meta-level.
\end{theorem}

\section{Guarded Interaction Trees}
\label{sec:gitrees}
The type of \emph{guarded interaction trees} (or \emph{\gitrees} for short) $\IT_{E}(A)$ is defined for a ground type $A$ and an effect signature $E$, as we explain below.
In \Cref{fig:gitrees_def} $\IT_{E}(A)$ is written down as a guarded datatype.\footnote{Formally, the datatype is given by a solution to a recursive equation, which we examine in \Cref{sec:recursion}.
But it is convenient to think of $\IT_E(A)$ as a recursive datatype in which every recursive occurrence is behind a $\latert$.}
\begin{figure}[t]
  \begin{align*}
    \mathsf{guarded\ type}\ \IT_E(A) &{}= \Rret : A \to \IT_E\\
                                     &\ \ALT \Fun : \latert (\IT_E(A) \to \IT_E(A)) \to \IT_E(A)\\
                                     &\ \ALT \Err : \Error \to \IT_E(A)\\
                                     &\ \ALT \Tau : \latert \IT_E(A) \to \IT_E(A) \\
                                     &\ \ALT \Vis : \prod_{\idx \in \mathtt{I}} \big( \Ins_{\idx}(\latert \IT_E(A)) \times (\Outs_{\idx}(\latert \IT_E(A)) \to \latert \IT_E(A))\big) \to \IT_E(A)
  \end{align*}
  
  \caption{Guarded datatype of interaction trees.}
  \label{fig:gitrees_def}
\end{figure}
The first constructor $\Rret$ says that any element $a$ of the ground type $A$ can be associated with a ``terminated'' guarded interaction tree $\Rret(a)$.
The second constructor $\Fun$ says that functions are also guarded interaction trees, and it is this constructor that allows us to model higher-order computations.
Since the function constructor contains a negative occurrence of $\IT_E(A)$ in its argument, we must put it under a $\latert$.
The third constructor $\Err$ represents an error state, or a stuck computation, which we take from some predefined set $\Error$ of errors.
We assume that it contains at least one element $\RunTime\in\Error$ representing a generic run-time error.
The fourth constructor $\Tau$ denotes a delayed computation, or a computation that is available ``later''.
We also write $\Tick : \IT_E(A) \to \IT_E(A)$ for the composition $\Tau \circ \Next$.
Then the term $\Tick(\alpha)$ represents a guarded interaction tree that takes a silent step to $\alpha$.
It satisfies the following rule for equality: 
$\Tick(\alpha) = \Tick(\beta) \provesIff \later (\alpha = \beta)$.

Finally, the last constructor $\Vis$ allows us to model effects.
The possible effects are given by the signature  $E = (\mathtt{I},\Ins_{-},\Outs_{-})$, where $\mathtt{I}$ is an indexing set on the meta-level (i.e. a set of operation names), and $\Ins$ and $\Outs$ are functors determining the arities of the operations.
That is $\Ins_{\idx}, \Outs_{\idx} : \Type \to \Type$ for $\idx \in \mathtt{I}$.
The $\Type$ argument in $\Ins_{\idx}$ and $\Outs_{\idx}$ is instantiated with $\IT_{E}$ itself, and is used for giving signatures to higher-order effects.
With this, the first argument to $\Vis_{\idx}$ is then the input for the operation, and the second parameter is a continuation which, given an arbitrary output of the operations, produces the remainder of the computation.
One way to visualize this is to think of $\Vis_{\idx}$ as a node in the tree, with the annotation $\Ins_{\idx}$ and having $\Outs_{\idx}$ many branches.

We refer to guarded interaction trees $\Rret(a)$ and $\Fun(f)$ as \gitree values, and write $\ITv_E(A) \subseteq \IT_E(A)$ for the set of values.
When quantifying over an indexing set, we implicitly coerce $E$ to $\mathtt{I}$, i.e. we write $\idx \in E$ to mean $\idx \in E.\mathtt{I}$.
Similarly we write $\Ins$ for $E.\Ins$ and $\Outs$ for $E.\Outs$ when the signature $E$ is clear from the context.
When the signature $E$ is obvious or unimportant we simply write $\IT(A)$ for $\IT_E(A)$ and $\ITv(A)$ for $\ITv_E(A)$.

Let us demonstrate the syntax of guarded interaction trees with some running examples of effects.
\begin{example}[Input/output on a tape]
  \label{ex:io_sig}
  Suppose we want to model two effectful operations, for reading a number from STDIN and for writing an output on STDOUT.
  We will model them as guarded interaction trees $\IT_{E_{io}}(\Tunit + \Tnat)$, where $\Tunit = \{()\}$ is the unit type and
  \begin{align*}
    E_{io} &\eqdef \{\mathtt{input}, \mathtt{output}\} &\\
    \Ins_{\mathtt{input}}(X) &\eqdef \Tunit & \Outs_{\mathtt{input}}(X) &\eqdef \Tnat\\
    \Ins_{\mathtt{output}}(X) &\eqdef \Tnat & \Outs_{\mathtt{ouput}}(X) &\eqdef \Tunit
  \end{align*}
  
  We write $\INPUT$ and $\OUTPUT(n)$ for the \gitrees
  \begin{align*}
    \INPUT \eqdef \Vis_{\mathtt{input}}((), \Lam n. \Next(\Rret(\inr(n)))) &&
 \OUTPUT(n) \eqdef \Vis_{\mathtt{output}}(n, \Lam x. \Next (\Rret(\inl(()))))
  \end{align*}
  Here we use $\inl(()) : \Tunit + \Tnat$ as a ``dummy'' value, since we do not care about the return value of $\OUTPUT$.
\end{example}
The operations $\INPUT$ and $\OUTPUT$ above are represented as \gitrees $\IT_{E_{io}}(\Tunit + \Tnat)$.
However, the exact ground type is not important, as long as it contains the unit $\Tunit$ and the natural numbers $\Tnat$.
As such, we assume that we can write down operations like $\INPUT$ and $\OUTPUT$ as \gitrees $\IT_{E_{{io}}}(A)$ where $A \simeq \Tunit + \Tnat + B$ for some type $B$.
We return again to this point in \Cref{sec:subeffects}, but for now we assume that we always pick a ground type $A$ that is ``large enough'' to represent all the ground values that we need.

\begin{example}[Higher-order store]
  \label{ex:store_sig}
  We can model higher-order store with the following signature.
  \begin{align*}
    E_{store} &\eqdef \{\mathtt{alloc}, \mathtt{read}, \mathtt{write}, \mathtt{dealloc}\} &\\
    \Ins_{\mathtt{alloc}}(X) &\eqdef X & \Outs_{\mathtt{alloc}}(X) &\eqdef \Loc\\
    \Ins_{\mathtt{read}}(X) &\eqdef \Loc & \Outs_{\mathtt{read}}(X) &\eqdef X\\
    \Ins_{\mathtt{write}}(X) &\eqdef \Loc \times X & \Outs_{\mathtt{write}}(X) &\eqdef \Tunit\\
    \Ins_{\mathtt{dealloc}}(X) &\eqdef \Loc & \Outs_{\mathtt{dealloc}}(X) &\eqdef \Tunit
  \end{align*}
  where $\Loc$ is a countable type of locations/pointers.
  We write $\ALLOC$, $\READ$, $\WRITE$, and $\DEALLOC$ for the following \gitrees:
  \begin{align*}
    \ALLOC(\alpha : \IT(A),k : \Loc \to \IT(A)) \eqdef{}& \Vis_{\mathtt{alloc}}(\Next(\alpha), \Next \circ k) \\
    \READ(\loc : \Loc) \eqdef{}& \Vis_{\mathtt{read}}(\loc, \Lam x. x)\\
    \WRITE(\loc : \Loc, \alpha : \IT(A)) \eqdef{}& \Vis_{\mathtt{write}}((\loc, \Next(\alpha)), \Lam x. \Next(\Rret(\inj(()))))\\
    \DEALLOC(\loc : \Loc) \eqdef{}& \Vis_{\mathtt{dealloc}}(\loc, \Lam x. \Next(\Rret(\inj(()))))
  \end{align*}
  Here we assume that the ground type $A$ is isomorphic to $\Tunit + B$ for some $B$, with the injection $\inj : \Tunit \to A$.
\end{example}

\subsection{Recursion Principle for Guarded Interaction Trees}
\label{sec:recursion}
In order to write programs that eliminate \gitrees, i.e. programs of type $\IT_{E}(A) \to P$, we need to come up with a suitable recursion principle.
Recursion principles for inductive datatypes usually follow from the initiality principles of the defined datatypes.
However, the type of \gitrees is not purely inductive, as it has mixed-variance recursive occurrences, and the corresponding recursion principle should reflect that.
To understand the necessary recursion principle we need to understand first how the \gitrees are defined.
The definition at the beginning of this section presents \gitrees as a guarded datatype, but how should such a datatype be constructed?
In the type theory, the type $\IT_E(A)$ is given as the solution to the following guarded equation:
\begin{align}
  \label{eq:gitree_def}
  \begin{split}
    \IT_E(A) \simeq{}& A + \latert [\IT_E(A) \to \IT_E(A)] + \Error + \latert \IT_E(A)  +{}\\
& \Sigma_{\idx \in E} \big( \Ins_{\idx}(\latert \IT_E(A)) \times (\Outs_{\idx}(\latert \IT_E(A)) \to \latert \IT_E(A))\big)
  \end{split}
\end{align}
The isomorphism is witnessed by the pair of functions $(\unfoldIT, \foldIT)$,
and the constructors we presented at the beginning of the section are compositions of injections and $\foldIT$.
Since \Cref{eq:gitree_def} contains recursive occurrences with mixed variance, we cannot use the usual recursion principle for inductive data type.
Instead, we employ a mixed initial-algebra/final-coalgebra principle, following \cite{Freyd:1991,Pitts:1996}. 
To understand it better, we first write out the bi-functor, where the fixed point corresponds to the type of \gitrees:\footnote{For a detailed category-theoretic treatment, see \cite{BirkedalEtAl:2010a}.}
\begin{align*}
  F(X,Y) \simeq{}& A + \latert [X \to Y] + \Error + \latert Y  +{}
   \Sigma_{\idx \in E} \big( \Ins_{\idx}(\latert Y) \times (\Outs_{\idx}(\latert X) \to \latert Y)\big)
\end{align*}
Here the bi-functor $F(-,-)$ is contravariant in the first argument and covariant in the second one.
The bi-algebra corresponding to the type of \gitrees is given by the $(\foldIT,\unfoldIT)$ pair:
\[
  \xymatrix{
    F(\IT,\IT) \ar@<-0.5ex>[r]_(0.6){\foldIT} &
    \IT \ar@<-0.5ex>[l]_(0.35){\unfoldIT}
  }
\]
where we write $\IT$ as a shorthand for $\IT_{E}(A)$.
The recursion principle that we are looking for then states that this bi-algebra is both initial and terminal.
That is, for any other bi-algebra $(P,f,g)$ we have unique maps $h$ and $k$ such that the following diagram commutes:
\[
  \xymatrix@C=+10pc{
    F(P,P) \ar@<-0.5ex>[d]_{f} \ar@<0.5ex>[r]^{F(k,h)}
         & F(\IT,\IT) \ar@<-0.5ex>[d]_{\foldIT} \ar@<0.5ex>[l]^{F(h,k)} \\
    P \ar@<-0.5ex>[u]_{g} \ar@<0.5ex>[r]^{h} & \ar@<0.5ex>[l]^{k} \IT \ar@<-0.5ex>[u]_{\unfoldIT}
  }
\]
That is, in order to construct a function $k : \IT \to P$, one has to provide
the ``unfolding'' $P \to F(P,P)$, as well as
functions
$A \to P$, $\Error \to P$, $\latert P \to P$, $\latert (P \to P) \to P$, and $\prod_{\idx\in E} \Ins_{\idx}(\latert P) \to (\Outs_{\idx}(\latert P) \to \latert P) \to P$.
This alone would allow us to \emph{iterate} over \gitrees.
However, we would run intro trouble if we wish to write a primitive-recursive style function.
For example, we might wish to write a destructor function $k$ such that $k(\alpha)$ returns $\alpha$ if $\alpha$ itself is a function $\Fun(f)$, and $\Err(\RunTime)$ otherwise.
We cannot do so with the scheme outlined above, since in the recursive call we don't have access to the original argument, only to the result of applying recursion to the argument.
This is similar to how the iteration scheme $B \to (B \to B) \to \mathbb{N} \to B$ for natural numbers does not allow us to (easily) write the predecessor function $p : \mathbb{N} \to \mathbb{N}$ satisfying $p(0) = 0$ and $p(n+1) = n$ if we pick $B = \mathbb{N}$.

\paragraph{Recursion from iteration on inductive types.}
Let us then look at how to solve the issue of defining primitive recursive functions on inductive types using initiality.
Suppose the function $p: \mathbb{N}\to B$ that we want to construct  is defined by equations $p(0) = p_{1}$ and $p(n+1)=p_{2}(n, p(n))$.
Then we can obtain this function $p$ using the following trick:
instead of eliminating $\mathbb{N}$ into $B$ using initiality, we eliminate it into $\mathbb{N} \times B$, in such a way that the induced map $\mathbb{N} \to \mathbb{N}\times B$ is the identity on the first component.
More concretely, suppose we have maps $p_1 : 1 \to B$ and $p_2 : \mathbb{N}\times B \to B$, forming together the equations for primitive recursion.
Then we construct an $\mathbb{N}$-algebra over $\mathbb{N} \times B$ as
\[
  \xymatrix@+3pc{
    1 + (\mathbb{N} \times B) \ar[r]^{[\langle 0, p_1 \rangle, \langle S, p_2 \rangle]} &
    \mathbb{N} \times B
  }
\]
where $S : \mathbb{N} \to \mathbb{N}$ is the successor function.
The initiality of $\mathbb{N}$ will then induce the unique map $p : \mathbb{N} \to \mathbb{N} \times B$, which, when composed with projection $\mathbb{N}\times B \to B$, determines the recursive function given by the clauses $p_1$ and $p_2$.

\paragraph{Recursion/corecursion from mixed-variance types.}
Dually, for coinductive datatypes we can obtain a form of primitive coinduction from coiteration by using coproducts.
In our case we have a datatype with mixed variance, and we use coproducts for the negative occurrences and products for the positive ones.
That is, in order to eliminate $\IT$ into a type $P$ we will assume an unfolding $P \to F(P,P)$, and a folding $F(\IT+P, \IT\times P) \to P$.
More concretely:
\begin{definition}[Recursion/corecursion principle]
  In order to define a pair of maps $\xymatrix{P \ar@<0.5ex>[r]^{h} & \IT \ar@<0.5ex>[l]^{k}}$, one has to define maps
  \begin{itemize}
  \item $h_u : P \to A  + \latert [P \to P] + \Error + \latert P + \Sigma_{\idx \in E} \big( \Ins_{\idx}(\latert P) \times (\Outs_{\idx}(\latert P) \to \latert P)\big)$;
    
  \item $k_{\Rret} : A \to P$;
  \item $k_{\Fun} : \latert \big((\IT + P) \to (\IT \times P)\big) \to P$;
  \item $k_{\Err} : \Error \to P$;
  \item $k_{\Tau} : \latert (\IT \times P) \to P$;
  \item $k_{\Vis} : \prod_{\idx : E} \Ins_{\idx}(\latert(\IT \times P)) \to (\Outs_{\idx}(\latert(\IT + P)) \to
    \latert (\IT \times P)) \to P$.
  \end{itemize}
  The resulting maps $(h,k)$ will then satisfy the following \emph{computational rules}:
  \begin{itemize}
  \item $k(\Rret(a)) = k_{\Rret}(a)$;
  \item $k(\Fun(f)) = k_{\Fun}(\latert s(f))$
    where $s(f) = \langle \idfun, k \rangle\circ f \circ [\idfun,h]$;
  \item $k(\Err(e)) = k_{\Err}(e)$;
  \item $k(\Tau(t)) = k_{\Tau}\big(\latert\langle \idfun, k\rangle(t) \big)$;
  \item $k(\Vis_{\idx}(x,k)) = k_{\Vis}(\idx,\ \Ins_{\idx}(\latert\langle\idfun,k\rangle)(x),
    \ \latert\langle\idfun,k\rangle \circ k \circ \Outs_{\idx}(\latert [\idfun,h]))$;
  \item plus equations for $h$.
  \end{itemize}
(Recall that we write $\latert s : \latert A  \to \latert B$ for a function $s : A \to B$.)
\end{definition}
This recursion/corecursion principle is constructed using guarded recursion, and can be used to define a large variety of combinators.
For example, we can write a generalization of the aforementioned function $k$ that returns its argument, if the argument is a function, and returns an error otherwise.

Using the recursion principle we can define a function $\getfun(\alpha : \IT, f : \latert(\IT \to \IT) \to \IT)$ satisfying the computational rules
  \begin{itemize}
  \item $\getfun(\Rret(a),f) = \Err(\RunTime)$;
  \item $\getfun(\Fun(g),f) = f(g)$;
  \item $\getfun(\Err(e),f) = \Err(e)$;
  \item $\getfun(\Tau(t),f) = \Tau(\latert \getfun(t,f))$ and $\getfun(\Tick(\alpha),f) = \Tick(\getfun(\alpha),f)$;
  \item $\getfun(\Vis_i(x,k),f) = \Vis_i(x, \latert \getfun(-,f) \circ k)$.
  \end{itemize}
In the next section we will see how to use $\getfun$ to define an application function $\APPs{\alpha}{\beta}$ for applying a GITree function $\alpha$ to a GITree argument $\beta$. 
  
In the rest of the paper we will define other operations on \gitrees using just the computational rules, with the understanding that we can write down the explicit recursor for any such set of equations.
Interested readers are referred to the Coq formalization for the details.

\subsection{Programming with \gitrees}
\label{sec:programming}
Using the recursion principle we can define operations on \gitrees that correspond to common programming constructs.
For example, with $\getfun$ we can write a function $\APPl(\alpha,\beta)$, which applies $\alpha$ to $\beta$ if $\alpha$ is a function, and returns $\Err(\RunTime)$ otherwise.
\begin{equation*}
  \APPl(\alpha, \beta) = \getfun(\alpha,\Lam g.\Tau(\latert g(\beta))).
\end{equation*}
This operation gives us ``call-by-name'' application, in the sense that it satisfies $$\APPl(\Fun(\Next(g)),\beta) = \Tick(g(\beta))$$ for $g : \IT_{E}(A) \to \IT_{E}(A)$ for any argument $\beta$.
In particular, it invokes the underlying function $g$ even if the argument $\beta$ is a $\Tick$ or an effect, without evaluating the argument first.
In order to define a ``call-by-value'' application, we compose $\APPl$ with the following operation.

The function $\getval(\alpha, f)$ recurses into its argument, looking under $\Tick$s and $\Vis$'s, until it reaches either a function or a ground type (i.e. a value from $\ITv_{E}(A)$), after which it applies the function $f : \IT_{E}(A) \to \IT_{E}(A)$ to it:
\begin{align*}
  \getval(\Rret(a),f) & = f(\Rret(a)) &
     \getval(\Fun(g),f) &= f(\Fun(g))\\
  \getval(\Err(e),f) &= \Err(e) &
     \getval(\Tick(\alpha),f) &= \Tick(\getval(\alpha,f)) \\
  \getval(\Vis_i(x,k),f) &= \Vis_i(x, \latert \getval(-,f) \circ k)
\end{align*}
As syntactic sugar, we write $\LET x=\alpha IN \beta(x)$ for $\getval(\alpha, \Lam x.\beta(x))$.

Now we can define the ``call-by-value'' application $\APPs{\alpha}{\beta}$
as $\getval\left(\beta, \Lam \beta_v. \APPl(\alpha, \beta_v)\right)$.
This strict application then satisfies the following computational rules
\begin{align*}
  \APPs{\alpha}{\Tick(\beta)} &= \Tick(\APPs{\alpha}{\beta})
  & \APPs{\alpha}{\Vis_i(x,k)} &= \Vis_i(x, \Lam y. \APPsl{\Next(\alpha)}{k\ y})\\
  \APPs{\Tick(\alpha)}{\beta_v} &= \Tick(\APPs{\alpha}{\beta_v})
  & \APPs{\Vis_i(x,k)}{\beta_v} &= \Vis_i(x, \Lam y. \APPsl{k\ y}{\Next(\beta_v)})\\
  \APPs{\Fun(\Next(g))}{\beta_v} &= \Tick(g(\beta_v))
  & \APPs{\alpha}{\beta} &= \Err(\RunTime) \mbox{ in other cases }
\end{align*}
Where $\APPsl{-}{-}$ is the lifting of $\APPs{-}{-}$ to $\latert \IT_{E}(A) \to \latert \IT_{E}(A) \to \latert \IT_{E}(A)$, and $\beta_v \in \ITv_{E}(A)$ is either $\Rret(a)$ or $\Fun(g)$.
The application function $\APPs{-}{-}$ not only simulates strict application, but it also fixes a right-to-left evaluation order of effects and computation steps.

One can see that there are common properties for the computational rules between
$\getfun(-,f)$, $\APPl(-,\beta)$, $\APPs{\alpha}{-}$, and $\APPs{-}{\beta_v}$ (where $\beta_v \in \ITv_{E}(A)$): they all preserve ticks, effects, and errors.
Functions that have these preservation properties are called \emph{homomorphisms} of \gitrees and will play an important role in later sections.
\begin{definition}
  \label{def:homomorphism}
  A function $f : \IT_{E}(A) \to \IT_{E}(A)$ is a \emph{homomorphism}, written as $f \in \Hom$, if it satisfies the following equations:
  \begin{itemize}
  \item $f(\Err(e)) = \Err(e)$;
  \item $f(\Tick(\alpha)) = \Tick(f(\alpha))$;
  \item $f(\Vis_i(x,k)) = \Vis_i(x, \latert f \circ k)$
  \end{itemize}
  As expected from the name, the identity function is a homomorphism and composition of two homomorphisms is a homomorphism.
This notion of homomorphism is inspired by the one in \cite{Hoshino:2012}.
It follows from the definition, that in order to define a homomorphism it suffices to define its action on \gitree values.
\end{definition}

\paragraph{Programming with \gitrees and natural numbers.}
In the remainder of this paper we work with a lot of examples involving programming with natural numbers (as an illustrative ground type).
It is then useful to assume in the remainder of this paper that the ground type $A$ in any type $\IT_{E}(A)$ of \gitrees is ``large enough'' to contain natural numbers, and the unit type.
That is, we assume that $A \simeq \Tunit + \Tnat + \dots$, and we simply write $\Rret(n)$ for $\Rret(\inj(n))$ and $\Rret(())$ for $\Rret(\inj'(()))$ (for appropriate injections $\inj : \Tnat \to A$ and $\inj' : \Tunit \to A$).
We will also abbreviate $\IT_{E}(A)$ as $\IT$ or $\IT_{E}$ when $A$ is generic as above or is clear from the context.

In \Cref{fig:lambda_ops} we summarize the operations on \gitrees that we define using recursion and other functions.
The computational rules described in \Cref{fig:lambda_ops} are only for the base cases; the other computational rules follow from the fact that those operations are homomorphisms.
Concretely, we have the following operations.
The $\getnat$ function extracts a natural number from a \gitree and applies the function $f : \Tnat \to \IT$ to it.
It is a homomorphism in the first argument.
If it encounters a function $\Fun(g)$ or a different ground value $\Rret(b)$, then it returns an error.
The $\IF$ operations test whether the first argument is zero, and picks the appropriate branch.
The function $\IF(-,\alpha_1,\alpha_2)$ is a homomorphism.
Similarly, if the first argument is not a natural number then $\IF$ returns an error.
The $\NATOP_f$ operation applies the binary function $f$ to its integer arguments, returning an error on all the other values.
The maps $\NATOP_f(\alpha,-)$ and $\NATOP_f(-,\beta_v)$ are homomorphisms for $\beta_v \in \ITv$.
The $\alpha\SEQ \beta$ is a sequencing operation: it puts all the effects and ticks in $\alpha$ before the effects and ticks in $\beta$.
This is witnessed by the fact that $(-) \SEQ \alpha$ is a homomorphism.
The $\WHILE{\alpha}{\beta}$ represents a while loop with the conditional $\alpha$ and the body $\beta$; it is defined using guarded recursion, and is equal to its one-step unfolding using the $\IF$ construct.

\begin{figure}[t]\small
  \begin{mathpar}
    \infer{n \in \Tnat}{\getnat(\Rret(n), f) = f(n)}
    \and
    \infer{b \notin \Tnat}{\getnat(\Rret(b), f) = \Err(\RunTime)}
    \and
    \getnat(\Fun(g), f) = \Err(\RunTime)
    \and
    \getfun(\Rret(a), f) = \Err(\RunTime)
    \and
    \getfun(\Fun(g), f) = f(g)
    \and
    \IF(\Rret(0),\alpha_1,\alpha_2) = \alpha_1
    \and
    \infer{n > 0}
    {\IF(\Rret(n),\alpha_1,\alpha_2) = \alpha_2}
    \and
    \IF(\Fun(f),\alpha_1,\alpha_2) = \Err(\RunTime)
    \and
    \infer{b \notin \Tnat}{\IF(\Rret(b),\alpha_1,\alpha_2) = \Err(\RunTime)}
    \and
    \infer{n_1, n_2 \in \Tnat
    }
    {\NATOP_f(\Rret(n_1),\Rret(n_2)) = \Rret(f(n_1,n_2))}
    \and
    \infer{\alpha_v \mbox{ or } \beta_v \mbox{ are not $\Rret(n)$}}{\NATOP_f(\alpha_v,\beta_v) = \Err(\RunTime)}
    \and
    \beta_v \SEQ \alpha = \alpha
\and
    \WHILE{\alpha}{\beta} = \IF\big(\alpha, (\beta \SEQ \Tick(\WHILE{\alpha}{\beta})), \Rret(())\big)
  \end{mathpar}
  \caption{Programming operations on \gitrees.}
  \label{fig:lambda_ops}
\end{figure}

Let us look at some example programs that we can write using the operations we have defined.
\begin{example}[Factorial]
\label{ex:factorial}
  In the first example, we have a factorial function that we implement using the store operations (\Cref{ex:store_sig}).
\begin{align*}
  \fact(n) \eqdef{}& \ALLOC\left(\Rret(1), \Lam \acc. \ALLOC(\Rret(n), \Lam \ell.\factbody(\acc,\ell) \SEQ \READ(\acc))\right)\\
  \factbody(\acc,\ell) \eqdef{} &
  \WHILE{\READ(\ell)}
  {\\ & \quad
  \begin{aligned}[t]
    &\LET i = \READ(\ell) IN\\
    &\LET r = \NATOP_{\times}(i, \READ(\acc)) IN \\
    &\LET i = \NATOP_{-}(i,\Rret(1)) IN\\
    &\WRITE(\acc, r) \SEQ  \WRITE(\ell, i)
  \end{aligned}
  }
\end{align*}
The program $\factbody$ computes the factorial of the number stored in the location $\ell$ using an intermediate location $\acc$ for the accumulated result.
The complete program $\fact$ then allocates the needed references and runs $\factbody$ before reading off the result from the location $\acc$.
\end{example}

\begin{example}[Encoding of pairs]
  Our definition of \gitrees does not include arbitrary algebraic datatypes, like pairs or sums.
We can, however, encode them using a Church-style encoding.
We write $(\alpha, \beta) :\IT$ for the guarded interaction tree
$$
\LET y=\beta IN \LET x=\alpha IN \Fun(\Next(\Lam f. \APPs{\APPs{f}{x}}{y})).
$$
Note that $(\alpha_v,\beta_v)$ is a \gitree value whenever $\alpha_v$ and $\beta_v$ are.
Furthermore, $(\alpha, -)$ and $(-, \beta_v)$ are homomorphisms.
We then define the projection functions as
$$
\Proj{1}(\alpha) = \APPs{\alpha}{\Fun(\Next(\Lam a.\Fun(\Next(\Lam b.a))))}
\qquad
\Proj{2}(\alpha) = \APPs{\alpha}{\Fun(\Next(\Lam a.\Fun(\Next(\Lam b.b))))}.
$$
The projection functions then satisfy the following computational rules:
$$
\Proj{1}(\alpha_v,\beta_v) = \Tick^{3}(\alpha_v)
\qquad \Proj{2}(\alpha_v,\beta_v) = \Tick^{3}(\beta_v).
$$
We can use similar style encodings to represent other algebraic datatypes as guarded interaction trees.
\end{example}

\section{Reification of Effects and Reductions of \gitrees}
\label{sec:reductions}
\gitrees allow us to conveniently write down and combine various effects.
But in order to reason about the effects we also need a way of giving them meaning.
In this section we establish a way of \emph{reifying} effects of \gitrees and use reification to define \emph{reductions} of \gitrees, which explain how computations represented by \gitrees reduce.

In order to interpret stateful effects we assume that we have a type $\stateO$, and each effect is interpreted using the state monad with a function:
\[
  r : \prod_{\idx \in E} \Ins_{\idx}(\latert \IT_E) \times \stateO \to \optionO(\Outs_{\idx}(\latert\IT_E) \times \stateO). 
\]
We call a tuple $(E, \stateO, r)$ a \emph{reifier} for the effects $E$.
Assuming we have such a reifier, we write a function $\reify : \IT \times \stateO \to \IT \times \stateO$ that satisfies
\begin{mathpar}
  \infer
  {r_i(x,\sigma) = \Some(y, \sigma') \and k\ y = \Next(\beta)}
  {\reify(\Vis_i(x, k), \sigma) = (\Tick(\beta), \sigma')}
  \and
  \infer
  {r_i(x,\sigma) = \None}
  {\reify(\Vis_i(x, k), \sigma) = (\Err(\RunTime), \sigma)}
\end{mathpar}

\begin{example}[Reification for the input/output operations \Cref{ex:io_sig}]
  \label{ex:io_reify}
  We take the state $\stateO$ to be a pair of two lists of natural numbers, corresponding to input and output tapes.
  The reifier is defined as
  \begin{align*}
    r_{\mathtt{input}}((),(n\vec{n},\vec{m})) &= \Some(n, (\vec{n},\vec{m}))
    & r_{\mathtt{input}}((),(\epsilon,\vec{m})) &= \None\\
    r_{\mathtt{output}}(x,(\vec{n},\vec{m})) &= \Some((), (\vec{n},x\vec{m}))
  \end{align*}
\end{example}
\begin{example}[Reification for the higher-order store operations \Cref{ex:store_sig}]
  \label{ex:store_re}
  For higher-order store we take $\stateO$ to be the type of finite partial maps
  $\Loc \fpfn \latert \IT$.
  The reifier function is defined in the expected way:
  \begin{align*}
    r_{\mathtt{alloc}}(\alpha, \sigma) &= \Some(\loc, \sigma[\loc \mapsto \alpha]) \mbox{\qquad where $\loc$ is the smallest location not present in $\sigma$}\\
    r_{\mathtt{read}}(\loc, \sigma) &=
       \begin{cases}
         \Some(\alpha,\sigma) & \mbox{ if } \sigma(\loc) = \alpha\\
         \None & \mbox{ otherwise}
       \end{cases}\\
    r_{\mathtt{write}}((\loc,\beta), \sigma) &=
       \begin{cases}
         \Some((),\sigma[\loc\mapsto\beta]) & \mbox{ if $\sigma(\loc)$ is defined}\\
         \None & \mbox{ otherwise}
       \end{cases}\\
    r_{\mathtt{dealloc}}(\loc, \sigma) &=
       \begin{cases}
         \Some((),\sigma\setminus\{\loc\}) & \mbox{ if $\sigma(\loc)$ is defined}\\
         \None & \mbox{ otherwise}
       \end{cases}
  \end{align*}
\end{example}

\paragraph{From reification to reductions.}
Using the reification function we can formulate the reduction relation on interaction trees.
The \emph{internal} reduction relation $\istep : (\IT \times \stateO) \to (\IT \times \stateO) \to \Prop$:
\[
  (\alpha,\sigma) \istep (\beta,\sigma') \eqdef
  \big(\alpha = \Tick(\beta) \wedge \sigma = \sigma' \big)
  \vee \big(\Exists i\,x\,k. \alpha = \Vis_i(x,k)
  \wedge \reify(\alpha,\sigma) = (\Tick(\beta),\sigma')\big)
\]
Intuitively, a reduction of \gitrees corresponds to either stripping away one computational step, or to reifying an effect.
We consider an (annotated) transitive closure of the reduction  relation:
\begin{align*}
  (\alpha, \sigma) \istep^0 (\beta, \sigma') \eqdef{}&
   \alpha = \beta \wedge \sigma = \sigma' \\
  (\alpha, \sigma) \istep^{n+1} (\beta, \sigma') \eqdef{}&
  \Exists \alpha_0,\,\sigma_0. (\alpha, \sigma) \istep (\alpha_0,\sigma_0)
          \wedge (\alpha_0,\sigma_0) \istep^n (\beta,\sigma')
\end{align*}
We write $\istep^{\ast}$ for the reflexive transitive closure of the reduction relation.

\paragraph{Reductions and homomorphisms}
Homomorphisms (\Cref{def:homomorphism}) play an important role in the reduction relation, allowing us to compute the reductions more easily.
Specifically, homomorphisms preserve and reflect reductions:
\begin{lemma}
  \label{lem:hom_istep}
  Let $f$ be a homomorphism.
  Then $(\alpha,\sigma)\istep(\beta,\sigma')$ implies
  $(f(\alpha),\sigma)\istep(f(\beta),\sigma')$.
\end{lemma}
\begin{lemma}
  \label{lem:hom_istep_inv}
  Let $f$ be a homomorphism.
  If $(f(\alpha), \sigma)\istep(\beta',\sigma')$ then either
  \begin{itemize}
  \item $\alpha$ is a \gitree-value, or;
  \item there exists $\beta$ such that
    $(\alpha,\sigma)\istep(\beta,\sigma')$ and $\later(f(\beta) = \beta')$.
  \end{itemize}
\end{lemma}
These two lemmas suggest that homomorphisms play the role of evaluation contexts within the reduction relation $\istep$.
For example, if $(\alpha\SEQ \beta, \sigma)\istep(\delta, \sigma')$,
then either $\alpha$ was a value, or $(\alpha,\sigma)\istep(\alpha',\sigma')$
and $\later (\alpha'\SEQ \beta = \delta)$.

\paragraph{Continuation-independent reifiers.}
The reifiers that we consider here produce an output based on the input, but do not have direct access to the continuation.
The reification function $\reify$ just calls the continuation with the produced output.
This \emph{continuation-independence} is crucial for proving \Cref{lem:hom_istep_inv} (and the associated rule \ruleref{wp-hom} in separation logic in \Cref{sec:program_logic}).
Not all effects are continuation-independent, for example \texttt{call/cc} cannot be implemented this way.
In this paper, just like in \cite{XiaEtAl:2019},  we stick to working with continuation-independent effects, as it simplifies the separation logic and the reasoning principles, and we defer studying continuation-dependent effects to future work.

\section{Modeling a Higher-Order Effectful Programming Language}
\label{sec:iolang_model}
In this section we show how guarded interaction trees provide a model for a programming language with recursion, higher-order functions, and effects.
Specifically, we study a PCF-like higher-order programming language with input/output effects, give its interpretation into $\IT_{io}$ (see \Cref{ex:io_sig,ex:io_reify}), and show its soundness, i.e., that the interpretation agrees with the operational semantics.
The same approach applies to other classes of effects for which you can write operational semantics.

\paragraph{Syntax and operational semantics}
The syntax for the programming language, which we dub $\iolang$, consists of values and expressions:
\begin{align*}
\val \in \Val \bnfdef{}&
  n \ALT
  \Rec f \var = \expr \\
\expr \in \Expr \bnfdef{}&
  \var \ALT
  \val \ALT
  \If \expr then \expr_1 \Else \expr_2 \ALT
  \expr_1 (\expr_2) \ALT \expr_1 + \expr_2 \ALT \expr_1 - \expr_2 \ALT
  \Input \ALT \Output(\expr)
\end{align*}
where $n$ ranges over the set of natural numbers, and $f,\var$ range over the set $\Var$ of variables.

The operational semantics for $\iolang$ is given in \Cref{fig:iolang_opsem} as a small-step reduction relation on the configurations $\Expr \times \stateO$, where $\stateO$ is a pair of lists as in \Cref{ex:io_reify}.
\begin{figure}[t]\small
\begin{mathpar}
  \axiomH{red-beta}
  {((\Rec f \var = \expr)\ \val, \sigma) \step (\subst{\subst{\expr}{\var}{\val}}{f}{\Rec f \var = \expr},\sigma)}
  \and
  \inferH{red-natop}{n_1,n_2 \in \mathbb{N} \and \oplus \in \{+,-, \times, \dots\} \and n_1 \oplus n_2 = n}  
  {(n_1 \oplus n_2, \sigma) \step (n, \sigma)}
  \and
  \axiomH{red-if-false}
  {(\If 0 then \expr_1 \Else \expr_2, \sigma) \step (\expr_2, \sigma)}
  \and
  \inferH{red-if-true}{n \in \mathbb{N} \and n > 0}
  {(\If n then \expr_1 \Else \expr_2, \sigma) \step (\expr_1, \sigma)}
  \and
  \inferH{red-input}{}
  {(\Input, (n'\vec{n},\vec{m})) \step (n', (\vec{n},\vec{m}))}
  \and
  \inferH{red-output}{}
  {(\Output(m), (\vec{n},\vec{m})) \step (0, (\vec{n},m'\vec{m}))}
  \and
  \inferH{red-ectx}
  {(\expr_1,\sigma_1) \step (\expr_2, \sigma_2)}
  {(\fill\elctx[\expr_1],\sigma_1)\step (\fill\elctx[\expr_1],\sigma_2)}
\end{mathpar}
\caption{Small-step operational semantics for $\iolang$.}
\label{fig:iolang_opsem}
\end{figure}
The reductions are defined, following \cite{felleisen:hieb:1992}, using \emph{evaluation contexts} $\elctx \in \Ectx$, given as:
\begin{align*}
  \elctx \in \Ectx \bnfdef \ctxhole \ALT \Output(\elctx) \ALT \If \elctx then \expr_1 \Else \expr_2
  \ALT \expr\ \elctx \ALT \elctx\ \val \ALT \expr \oplus \elctx \ALT \elctx \oplus \val
\end{align*}
By $\fill\elctx[\expr]$ we denote the result of replacing the hole $\ctxhole$ in the context $\elctx$ with the expression $\expr$.
The evaluation contexts ensure the call-by-value right-to-left evaluation order of $\iolang$, as having a predefined evaluation order is important in the presence of effects.

\paragraph{Interpretation in guarded interaction trees}
We will interpret a closed program $\expr$ as an interaction tree $\Sem{\expr}:\IT_{\mathit{io}}(A)$.
The effects $\mathit{io}$ are those of \Cref{ex:io_sig,ex:io_reify}, and we assume that the ground type $A$ is ``large enough'' to have natural numbers.
For convenience, we drop the ground type and write simply $\IT_{\mathit{io}}$ for $\IT_{\mathit{io}}(A)$.

In order to provide a (compositional) denotational semantics we need to provide an interpretation not only for closed terms, but for open terms as well.
Given a set $\fv(\expr)=\{x_1, \dots, x_n\}$ of free variables of $\expr$, we define the interpretation $\Sem{\expr}_\rho:\IT_{\mathit{io}}$, where $\rho$ maps the free variables of $\expr$ to interaction trees.
The interpretation function is defined in \Cref{fig:iolang_interp}.
\begin{figure}[t]\small
\begin{mathpar}
\Sem{\var}_\rho = \rho(\var)
\and
\Sem{n}_\rho = \Rret(n)
\and
\Sem{\If \expr then \expr_1 \Else \expr_2}_\rho =
\IF(\Sem{\expr}_\rho,\Sem{\expr_1}_\rho,\Sem{\expr_2}_\rho)
\and
\infer{\oplus \in \{ +, -, \times, \dots \}}
{\Sem{\expr_1 \oplus \expr_2}_\rho = \NATOP_{\oplus}(\Sem{\expr_1}_\rho,\Sem{\expr_2}_\rho)}
\and
\Sem{\Input}_{\rho}=\INPUT
\and
\Sem{\Output(e)}_{\rho}=\getnat(\Sem{e}_\rho,\OUTPUT)
\and
\Sem{\expr_1\ \expr_2}_{\rho} = \APPs{\Sem{\expr_1}_\rho}{\Sem{\expr_2}_\rho}
\and
\Sem{\Rec f x = \expr}_{\rho} = \fix_{\IT}(\Lam (t:\latert \IT).
\Fun(\latert (\Lam \alpha\ v. \Sem{\expr}_{\rho[x\mapsto v][f \mapsto\alpha]})(t))).
\end{mathpar}
\caption{Semantic interpretation for $\iolang$.}
\label{fig:iolang_interp}
\end{figure}
The definition follows the standard notion of semantics for (untyped) $\lambda$-calculus, adjusted for effects and explicit recursion.
The interpretation of recursive functions $\Rec f x = \expr$ is defined using the guarded fixed pointed operation $\fix_{\IT} : (\latert \IT \to \IT) \to \IT$, and satisfies the following equality:
$$
\Sem{\Rec f x = \expr}_{\rho} = \Fun(\Next(\Lam v.\Sem{\expr}_{\rho[x\mapsto v][f\mapsto\Sem{\Rec f x = \expr}_{\rho}}))
$$
We show that the interpretation is sound:
\begin{theorem}[Soundness]
  \label{thm:iolang_soundness}
  If $(\expr_1,\sigma_1) \step (\expr_2,\sigma_2)$, then
         $(\Sem{\expr_1},\sigma_1)\istep^{\ast}(\Sem{\expr_2},\sigma_2)$.
\end{theorem}
We prove \Cref{{thm:iolang_soundness}} by induction on the $\step$-derivation.
The most interesting cases are for the reductions \ruleref{red-beta} and \ruleref{red-ectx}, which we now sketch. 
For the former, we need a substitution lemma, and for the latter we need to extend the interpretation to evaluation contexts.
\begin{lemma}[Substitution lemma]
  For any expression $\expr$ with a free variable $x$ we have
$$
\Sem{\subst{\expr}{x}{\expr'}}_{\rho} = \Sem{\expr}_{\rho[x\mapsto\Sem{\expr'}]}.
$$
\end{lemma}
\begin{proof}
  By induction on $\expr$, using L\"ob induction in the case of recursive functions.
\end{proof}
In order to handle  \ruleref{red-ectx} we provide the following auxiliary interpretation for evaluation contexts.
Each evaluation context $\elctx$ is interpreted as a \emph{homomorphism} $\Sem{\elctx}_\rho : \IT \to \IT$, such that
$\Sem{\fill{\elctx}[\expr]}_\rho = \Sem{\elctx}_\rho(\Sem{\expr}_\rho)$,
which together with \Cref{lem:hom_istep} implies the soundness of the \ruleref{red-ectx} reduction.

\section{Separation Logic over \gitrees}
\label{sec:program_logic}
In this section we define a separation logic as a program logic for guarded interaction trees.
We define a proposition $\wpre{\alpha}{\Phi}$ to denote that an interaction tree $\alpha$ is safe to reduce, and if $\alpha$ reduces to an interaction tree value $\beta_v$, then $\beta_v$ satisfies the postcondition $\Phi : \ITv \to \Prop$.

In this section we make use of the separation logic connectives of Iris, which we recall here:
\begin{align*}\small
  \prop \bnfdef{}& \dots \ALT \prop \ast \prop \ALT
                   \prop \wand \prop \ALT \upd \prop \ALT \Box \prop \ALT \knowInv{}{P} \ALT\dots
\end{align*}
For brevity, we only briefly recall the intuitive reading of these propositions and refer to \cite{DBLP:journals/jfp/JungKJBBD18} for details.
The proposition $P \ast Q$ says that the propositions $P$ and $Q$ hold over disjoint resources; the proposition $P \wand Q$ says that if we were to add any resources which satisfy $P$, then $Q$ would be satisfied.
The proposition $\upd P$ says that the current resources can be updated to satisfy $P$.
The proposition $\Box P$ states that $P$ holds \emph{persistently}, i.e., without asserting any resources.
Crucially, such propositions can be duplicated: $\Box P \vdash \Box P \ast \Box P$.
An example of a persistent proposition is the invariant proposition $\knowInv{}{P}$, which satisfies $\knowInv{}{P} \vdash \Box \knowInv{}{P}$.

\begin{figure}[t]\small
  \begin{mathpar}
    \inferH{wp-val}
    {\alpha \in \ITv \and \Phi(\alpha)}
    {\wpre{\alpha}{\Phi}}
    \and
    \inferH{wp-tick}
    {\later  \wpre{\alpha}{\Phi}}
    {\wpre{\Tick(\alpha)}{\Phi}}
    \and
    \inferH{wp-hom}
    {f \in \Hom \and \wpre{\alpha}{\Ret \beta_v. \wpre{f(\beta_v)}{\Phi}}}
    {\wpre{f(\alpha)}{\Phi}}
    \and
    \inferH{wp-reify}
    {\hasstate(\sigma) \and
      \reify(\Vis_i(x,k), \sigma) = (\Tick(\beta), \sigma')
      \and
      \later\big(\hasstate(\sigma') \wand \wpre{\beta}{\Phi} \big)}
    {\wpre{\Vis_i(x, k)}{\Phi}}
    \and
    \inferH{wp-upd}
    {\upd \wpre{\alpha}{\Ret \alpha_v. \upd \Phi(\alpha_v)}}{\wpre{\alpha}{\Phi}}
    \and
    \inferH{wp-mono}
    {\wpre{\alpha}{\Psi} \and \All \alpha_v.\Psi(\alpha_v) \wand \Phi(\alpha_v)}
    {\wpre{\alpha}{\Phi}}
\and
\inferH{wp-lam}
  {\wpre{\beta}{\Ret \beta_v.\later \wpre{f(\beta_v)}{\Phi}}}
  {\wpre{\APPs{\Fun(\Next(f))}{\beta}}{\Phi}}
\end{mathpar}
  \caption{Selected weakest precondition calculus rules.}
  \label{fig:wpre}
\end{figure}
Selected rules for the weakest precondition proposition $\wpre{\alpha}{\Phi}$ are given in \Cref{fig:wpre}.
The rule \ruleref{wp-val} states that to verify a value it suffices to check that the value satisfies the postcondition.
The rule \ruleref{wp-tick} states that in order to verify $\Tick(\alpha)$ it suffices to verify $\alpha$, under a later $\later$.
The rule \ruleref{wp-hom} states that in order to verify $f(\alpha)$ for a homomorphism $f$, it suffices to reduce $\alpha$ to a value $\alpha_v$, and then verify $f(\alpha_v)$.

The rule \ruleref{wp-reify} tells us how to deal with effects.
The rule uses the proposition $\hasstate(\sigma)$ which signifies the \emph{exclusive ownership} of the current state $\sigma$.
The use of separation logic is crucial in this case, as we do not want to allow duplicating that proposition.
The rule then states that in order to verify an effect, one has to provide the current state $\sigma'$ and the proof that the interaction tree with the effect reifies into some $\Tick(\beta)$.
Then, the user has to verify that the resulting $\beta$ reduces to a value satisfying the postcondition, under the assumption that the state has been updated to $\sigma'$.

The rule \ruleref{wp-upd} states that one can update ghost resources before and after reducing $\alpha$.
The rule \ruleref{wp-mono} states that one can always weaken the postcondition in $\wpre{\alpha}{\Phi}$.
Finally, \ruleref{wp-lam} is an example of a derived rule.
It combines the computational rule for function application of \gitrees, and rules \ruleref{wp-hom} and \ruleref{wp-tick}.
Let us look at an example derivation using these rules.
\begin{example}
Consider a $\iolang$ expression $(\Input + 1)$.
It is interpreted as the \gitree $\Sem{\Input + 1} = \NATOP_{+}(\INPUT, \Rret(1))$, for which we can prove the following specification:
  \begin{mathpar}\small
    \infer
    {\hasstate(n\vec{n},\vec{m}) \and
    \later(\hasstate(\vec{n},\vec{m}) \wand \Phi(\Rret(n+1)))}
    {\wpre{\NATOP_{+}(\INPUT, \Rret(1))}{\Phi}}
  \end{mathpar}
\end{example}
\begin{proof}
Note that $\NATOP_{+}(-,\Rret(1))$ is a homomorphism.
We apply \ruleref{wp-hom}, reducing our goal to:
  \[
    \wpre{\INPUT}{\Ret \beta_v. \wpre{\NATOP_{+}(\beta_v, \Rret(1))}{\Phi}}.
  \]
At this point we can use the assumption $\hasstate(n\vec{n},\vec{m})$
and the rule \ruleref{wp-reify}.
By the reifier of input/output effects, ${\reify(\INPUT, (n\vec{n},\vec{m})) = \Some(\Tick(\Rret(n)), (\vec{n},\vec{m}))}$, and we get the following goal:
  \[
    \later(\hasstate(\vec{n},\vec{m}) \wand \wpre{\Rret(n)}{\Ret \beta_v. \wpre{\NATOP_{+}(\beta_v, \Rret(1))}{\Phi}}).
  \]
Recall that we still have the assumption
  $
    \later(\hasstate(\vec{n},\vec{m}) \wand \Phi(\Rret(n+1))).
  $
  By the monotonicity of $\later$ we can remove the $\later$ modality both from the goal and the assumption.
  Since $\Rret(n)$ is a \gitree value, we can use \ruleref{wp-val} and reduce the goal to
  \[
    \wpre{\NATOP_{+}(\Rret(n), \Rret(1))}{\Phi}.
  \]
  By calculation, $\NATOP_{+}(\Rret(n), \Rret(1)) = \Rret(n+1)$, which is also a \gitree value.
  We can then apply \ruleref{wp-val} again to reduce the goal to $\Phi(\Rret(n+1))$, which follows from the assumption.
\end{proof}

We define the weakest precondition as a guarded recursive predicate, as is standard in Iris.
The weakest precondition then satisfies the following \emph{adequacy} and \emph{safety} theorem, the proof of which relies on the adequacy of Iris (\Cref{thm:iris_adequacy}).
\begin{theorem}
  \label{thm:wp_adequacy}
  Let $\alpha$ be an interaction tree and $\sigma$ be a state such that
  \[\hasstate(\sigma) \vdash \wpre{\alpha}{\Phi}
  \]
  is derivable for some meta-level predicate $\Phi$ (containing only intuitionistic logic connectives).
  Then for any $\beta$ and $\sigma'$ such that
  $(\alpha, \sigma) \istep^{\ast} (\beta,\sigma')$, one of the following two things hold:
  \begin{itemize}
  \item  \emph{(adequacy)} either $\beta \in \ITv$, and $\Phi(\beta)$ holds in the meta-logic;
  \item \emph{(safety)} or there are $\beta_1$ and $\sigma_1$ such that
    $(\beta,\sigma')\istep(\beta_1,\sigma_1)$
  \end{itemize}
  In particular, safety implies that $\beta \neq \Err(e)$ for any error $e \in \Error$.\footnote{While the weakest precondition that we presented in this section disallow any errors in guarded interaction tree, we will consider in \Cref{sec:comb_safety} that allows \emph{some} errors, at the user's discretion.
}
\end{theorem}

Finally, it is worth noting that separation logic/Iris is useful for reasoning about higher-order \gitrees even in the absence of effects, as demonstrated by the following example.
\begin{example}
  \newcommand{\ITER}{\mathsf{Iter}}
  Using guarded recursion, we can write down a \gitree $\ITER$ that satisfies the equation:
  $$
  \APPs{\APPs{\APPs{\ITER}{f}}{\alpha}}{\beta} = \IF\big(\alpha, \APPs{f}{(\APPs{\APPs{\APPs{\ITER}{f}}{\NATOP_{-}{}(\alpha,\Rret(1))}}{\beta})}, \beta\big).
  $$
  That is, $\APPs{\APPs{\APPs{\ITER}{f}}{\Rret(n)}}{\beta}$ computes the iterated application $\APPs{f^{n}}{\beta}$.
  We can give $\ITER$ the following higher-order specification:
  $$
  \infer
  {\wpre{\beta}{\Psi} \and \Box \All \beta_{v}. \Psi(\beta_{v}) \wand \wpre{\APPs{f}{\beta_{v}}}{\Psi}}
  {\wpre{\big(\APPs{\APPs{\APPs{\ITER}{f}}{\Rret(n)}}{\beta}\big)}{\Psi}}
  $$
  The specification says that if $\beta$ initially satisfies $\Psi$ and $f$ preserves $\Psi$ , then
  $\APPs{\APPs{\APPs{\ITER}{f}}{\Rret(n)}}{\beta}$ will also satisfy $\Psi$.
  The second premise is the specification of $f$, and it can be used multiple times in the proof.
  For that reason the that premise is behind the persistently modality $\Box$.

  It is also worth noting that while $\ITER$ itself does not use state, the function $f$ that we supply to it might as well use all sorts of effects internally, and our implementation and specification of $\ITER$ is oblivious to that.
\end{example}

\subsection{Domain-Specific Logic for Higher-Order Store}
\label{sec:store_logic}
Now we show how we can use the standard mechanisms in Iris to recover a fairly standard-looking separation logic for a programming language with references from the weakest precondition calculus presented above.
We use Iris's notion of \emph{higher-order ghost state} \cite{DBLP:conf/icfp/0002KBD16,DBLP:journals/jfp/JungKJBBD18} to provide the following logical interface for the higher-order store operations:
\begin{mathpar}\small
\inferH{wp-alloc}
{\heapctx \and \later \All \ell. \ell \mapsto \alpha \wand \wpre{k\ \ell}{\Phi}}
{\wpre{\ALLOC(\alpha, k)}{\Phi}}
\and
\inferH{wp-read}
{\heapctx \and \later \ell \mapsto \alpha \and
\later \big(\ell \mapsto \alpha \wand \wpre{\alpha}{\Phi}\big)}
{\wpre{\READ(\ell)}{\Phi}}
\and
\inferH{wp-write}
{\heapctx \and \later \ell \mapsto \alpha \and
\later \big(\ell \mapsto \beta \wand \Phi(\Rret(()))\big)}
{\wpre{\WRITE(\ell,\beta)}{\Phi}}
\and
\inferH{wp-dealloc}
{\heapctx \and \later \ell \mapsto \alpha \and
\later \Phi(\Rret(()))}
{\wpre{\DEALLOC(\ell)}{\Phi}}
\and
\infer{}{\heapctx\vdash \Box\heapctx}
\end{mathpar}
Here $\heapctx$ is a persistent proposition, which is part of the logical interface.

Thus our goal is to provide definitions of $\heapctx$ and $\ell \mapsto \alpha$ that allow us to derive the rules above.
The main challenge is that the resource $\hasstate(\sigma)$ provides a singular complete view of the state, without the ability to split it into local portions corresponding to individual locations.
That $\hasstate$ is not splittable by itself is not surprising -- it is an abstract representation of an arbitrary state for arbitrary effects, and there is no a priori way of splitting it.
However, for our specific state (a heap $\Loc \fpfn \latert \IT$) we know how to do the splitting.
What we need to do is to provide an alternative view of the state, amenable to splitting, and tie it together with the actual state of the $\hasstate$ predicate.

Our first step is then to provide a \emph{resource algebra} for this view of the state.
Following the standard practice of Iris, we use an authoritative resource algebra of the heap.
It contains two kinds of resources: the ``full heap'' $\authfull \sigma$ and the ``fragmental heap'' $\authfrag \sigma'$.
The fragmental heap is guaranteed to be a subheap of the full one $\sigma' \subseteq \sigma$.
Then we make the following definitions:
\begin{align*}
  \heapctx \eqdef{}& \knowInv{}{\Exists \sigma. \hasstate(\sigma) \ast \ownGhost{}{\authfull \sigma}}
&   \ell \mapsto \alpha \eqdef{}& \ownGhost{}{\authfrag [\ell \mapsto \Next(\alpha)]}
\end{align*}
The $\heapctx$ predicate is an invariant that says that the full view of the heap coincides with the actual state that we have as part of $\hasstate$, and the points-to predicate $\ell \mapsto \alpha$ states that $[\ell \mapsto \Next(\alpha)]$ is in the fragmental view of the heap.
Together those predicates imply that the actual state $\sigma$ maps $\ell$ to $\Next(\alpha)$, which is precisely what allows us to deduce the rules \ruleref{wp-alloc}, \ruleref{wp-read} and \ruleref{wp-write} from the rule \ruleref{wp-reify}.

As a simple example, we can use the rules for the store operations to verify the factorial program from \Cref{ex:factorial}.
We show the following specification: $\wpre{\fact(n)}{\Ret \beta_v. \beta_v = \Rret(!n)}$.
For this, we will use the intermediate lemma:
\begin{lemma}
  \label{lem:factbody}
  Under the assumptions $\heapctx$, $\acc \mapsto \Rret(m)$ and $\ell \mapsto \Rret(n)$, we have
$$\wpre{\factbody(\acc,\ell)}{\Ret \_. \acc \mapsto \Rret(m \times !n)}.$$
\end{lemma}
\begin{proposition}
\label{prop:factorial_spec}
  Under the assumption $\heapctx$ we have 
$$\wpre{\fact(n)}{\Ret \beta_v. \beta_v=\Rret(!n)}.$$
\end{proposition}
\begin{proof}
  We proceed by allocating the locations $\acc$ and $\ell$ symbolically using \ruleref{wp-alloc}, and then appeal to \Cref{lem:factbody}.
\end{proof}

As one can see, the logic that we recovered for the higher-order store effects is very close to a normal separation logic one would normally see for a programming language with a heap \cite{DBLP:journals/jfp/JungKJBBD18}.
Our logic, however, is amenable to extensions with other effects and programming language constructs.
Indeed, we explain how to obtain a logic for reasoning about different combined effects in \Cref{sec:subeffects}.
In the next section we show how to apply the separation logic to show computational adequacy of the model of $\iolang$.

\section{Computational Adequacy for $\iolang$}
\label{sec:adequacy}
In \Cref{sec:iolang_model} we constructed a compositional model of $\iolang$ in guarded interaction trees and proved that it is sound: if a $\iolang$ program $\expr$ terminates to a natural number $n$, then $\Sem{\expr}$ terminates to $\Rret(n)$.
In this section we show the other direction, known as computational adequacy in domain theory \cite{Plotkin:1977}, for the well-typed fragment of $\iolang$; the typing relation ($\Gamma \proves \expr : \type$) is given in \Cref{fig:iolang_typing}.
Computational adequacy is formally stated as the following theorem:
\begin{theorem}[Adequacy]
If $\ \vdash \expr : \Tnat$ and $(\Sem{\expr}, \sigma) \istep^{\ast}(\Rret(n),\sigma')$ then $(\expr,\sigma)\step^{\ast}(n,\sigma')$.
\end{theorem}
Computational adequacy is usually proved using logical relations between the syntax (terms of $\iolang$ in our case) and semantics (guarded interaction trees in our case).
Here we follow the recent practice \cite{DBLP:conf/popl/KrebbersTB17} of using the separation logic (see \Cref{sec:program_logic}) to define our logical relations model.
\begin{figure}[t]\small
  \begin{mathpar}
    \axiom{\var : \type ,\Gamma \proves \var : \type}
    \and
    \axiom{\Gamma \vdash n : \Tnat}
    \and
    \infer{\Gamma \vdash \expr_1 : \Tnat
      \and \Gamma \vdash \expr_2 : \Tnat \and \oplus \in \{-,+\}}
    {\Gamma \vdash\expr_1 \oplus \expr_2 : N}
    \and
    \infer
    {\Gamma \proves \expr_1 : \type_1 \to \type_2
      \and \Gamma \proves \expr_2 : \type_1}
    {\Gamma \proves \expr_1\ \expr_2 : \type_2}
    \and
    \infer{f : \type_1 \to \type_2, x : \type_1 , \Gamma \proves
      \expr : \type_1 \to \type_2
    }
    {\Gamma \proves \Rec f x = \expr : \type_1 \to \type_2}
    \and
    \infer{\Gamma \vdash \expr : \Tnat
    \and \Gamma \vdash \expr_1 : \type
      \and \Gamma \vdash \expr_2 : \type}
    {\Gamma \vdash \If \expr then \expr_1 \Else \expr_2 : \type}
    \and
    \axiom
    {\Gamma \vdash \Input : \Tnat}
    \and
    \infer{\Gamma \vdash e :\Tnat}
    {\Gamma \vdash \Output(e) :\Tnat}
  \end{mathpar}
  \caption{Typing rules for $\iolang$.}
  \label{fig:iolang_typing}
\end{figure}

We define a logical relation $\Gamma \models \alpha \refines \expr : \type$, relating a guarded interaction tree $\alpha$ and an expression $\expr$.
Here, $\expr$ is an open expression for which we have $\Gamma \proves \expr : \type$ while $\alpha$ is ``an open interaction tree'', i.e., a function of type $(\fv(\Gamma) \to \IT) \to \IT$.
As usual, we first define the relation over closed \gitrees and expressions, and then generalize it to the open case.
The logical relation, given in \Cref{fig:iolang_logrel}, is, as usual for call-by-value languages, is decomposed into an expression relation $\EE{\type}$ and a value relation $\VV{\type}$.
The expression relation simply states that related expressions should produce related values.
The value relation is defined by induction on the type in the standard way: values of base types should be equal while functions take related values to related expressions.
As values, once computed, can be used multiple times (cf.\ the logical relation in \Cref{sec:afflang_model}) the value relation is required to be persistent; hence the persistently modality $\Box$ in the value relation for functions.
In order to define the relation on open terms we define a relation for typing contexts $\VVCtx{\Gamma}$ which relates, point-wise, two substitutions respectively of the types $\fv(\Gamma) \to \ITv$ and $\fv(\Gamma) \to \Val$.
\begin{figure}[t]\small
  \begin{align*}
    \EE{\type}(\alpha,\expr) \eqdef{}& 
      \All \sigma. \hasstate(\sigma)\wand{} \wpre{\alpha}{\Ret \beta_v.
        \Exists v,\, \sigma'. (\expr,\sigma)\step^{\ast}(v,\sigma') \ast \VV{\type}(\beta_v, v) \ast \hasstate(\sigma')}\\
    \VV{\Tnat}(\beta_v,v) \eqdef{}& \Exists n \in \mathbb{N}. \beta_v = \Rret(n) \wedge v = n\\
    \VV{\type_1 \to \type_2}(\beta_v,v) \eqdef{}& 
      \Exists f. \beta_v = \Fun(f) \wedge \Box
      \big(\All \alpha_w,\, w. \VV{\type_1}(\alpha_w, w) \wand \EE{\type_2}(\APPs{\beta_v}{\alpha_w},v\ w)\big).\\
    \VVCtx{\Gamma}(\rho_1,\rho_2) \eqdef{}&\All (\var :\type) \in \Gamma. \VV{\type}(\rho_1(x), \rho_2(x))\\
    \Gamma \models \alpha \refines \expr : \type \eqdef{}&
      \All \rho_1,\ \rho_2. \VVCtx{\Gamma}(\rho_1, \rho_2) \implies
          \EE{\type}(\alpha(\rho_1), \expr[\rho_2])
  \end{align*}
  \caption{Logical relation for $\iolang$.}
  \label{fig:iolang_logrel}
\end{figure}
\begin{lemma}[Fundamental property]
  For any $\Gamma \vdash \expr : \type$, we have
  $\Gamma \models (\Lam \rho.\Sem{\expr}_{\rho}) \refines \expr : \type$.
\end{lemma}

Computational adequacy follows from the fundamental property together with the following Lemma which itself is a consequence of the soundness of the weakest precondition calculus (\Cref{thm:wp_adequacy}):
\begin{lemma}
  Suppose that $\ \models \alpha \refines \expr : \Tnat$.
  Then for any state $\sigma$,
  \begin{itemize}
  \item $if (\alpha,\sigma)\istep^{\ast}(\Rret(n),\sigma')$, then $(\expr,\sigma)\step^{\ast}(n,\sigma')$;
  \item if $(\alpha,\sigma)\istep^{\ast}(\beta,\sigma')$, then $\beta \neq \bot$.
  \end{itemize}
\end{lemma}

\section{Modular Reasoning About Combinations of Effects}
\label{sec:subeffects}
Because (guarded) interaction trees define effects abstractly, one of the main advantages is the ability to combine programs with different effects modularly in the same setting.
In this section we demonstrate how we achieve this for guarded interaction trees.

Given two signatures $E$ and $F$, with indexing sets $I$ and $J$, we say that $E$ is a subsignature of $F$, written as $\subEff{E}{F}$, if there is a mapping $\epsilon : I \to J$ such that $E.\Ins_i(X) \simeq F.\Ins_{\epsilon(i)}(X)$ and $E.\Outs_i(X) \simeq F.\Outs_{\epsilon(i)}(X)$ for any $i \in I$ and for any type $X$.
Here, $\simeq$ stands for isomorphism of types.

In regular interaction trees, a subsignature $\subEff{E}{F}$ induces an embedding $\IT_E \to \IT_F$ of interaction trees.
However, such an embedding is not possible for guarded interaction trees due to the mixed-variance definition:
a function $\IT_E \to \IT_E$ cannot be converted to a function $\IT_F \to \IT_F$ which takes a guarded interaction tree with a larger set of effects.

To achieve modularity we will instead work with an open-ended collection of effects which is large enough to embed all the effects that we need.
It is only at the ``top-level'', e.g., when applying the adequacy theorem, that we pick a concrete signature of effects.
For example, the precise type of the $\ALLOC$ function from \Cref{ex:store_sig} is the following type, for any $F$ such that $\subEff{E_{\mathit{store}}}{F}$:
\[
  \ALLOC : \IT_F \times (\Loc \to \IT_F) \to \IT_F
\]
(In practice, the function $\ALLOC$ is polymorphic in $F$ \emph{at the meta-level}, i.e., in Coq.)
With this, we can easily combine two programs with different collections of effects, assuming both of the programs are written in such an open ended manner; we just need to pick $F$ to be large enough to embed the effects of both programs.
For example, we can combine the factorial implementation from \Cref{ex:factorial} with input/output effects, to write a program that takes a natural number from the input, computes its factorial, and prints the result to the output:
$$\mathsf{fact\_io} \eqdef \getnat(\getnat(\INPUT,\fact), \OUTPUT).$$
The resulting program $\mathsf{fact\_io}$ has type $\IT_F$ for any $F$ such that $\subEff{E_{\mathit{store}}}{F}$ and $\subEff{E_{\mathit{io}}}{F}$.

\paragraph{Reifiers for modular effects}
Writing down programs with modular combinations of effects is not enough by itself: we also want to reason about the reification of effects modularly.
Suppose we write a program with effects $E$ as an \gitree $\IT_F$ with $\subEff{E}{F}$, and suppose that we have a reifier for $E$.
Recall that we defined a reifier for the effects $E$ to be a tuple $(E, \stateO, r : \prod_{\idx\in E} \Ins_{\idx}(\latert \IT_E) \times \stateO \to \optionO(\Outs_{\idx}(\latert\IT_E) \times \stateO)$.
However, if the state itself includes interaction trees, as in \Cref{ex:store_re}, we need also to make the state and the reifier parametric in the effects.
Therefore, instead of a fixed type $\stateO$ we consider a family of states $\stateO(X)$, and instead a single reifier function $r$ we consider a family of functions
\[
  \All X. \prod_{\idx \in E} \Ins_{\idx}(X) \times \stateO(X) \to \optionO(\Outs_{\idx}(X) \times \stateO(X)).
\]

In practice, we assume that the global state is the product of states of reifiers for sub-effects, in which each sub-effect acts only on its own part of the state.
Concretely, given a sequence $\vec{R} = (E_1,\stateO_1,r_1),\dots, (E_m,\stateO_m,r_m)$ of reifiers we define the \emph{global reifier} $R_G = (G, \stateO(-), r)$:
\begin{equation*}
  G.I \eqdef{} \sum_{1 \leq i \leq m} E_i.I
\qquad
  G.\Ins_{(i,j)}(X) \eqdef{} E_i.\Ins_j(X)
\qquad  G.\Outs_{(i,j)}(X) \eqdef{} E_i.\Outs_j(X)  
\end{equation*}
\vspace{-2em}
\begin{align*}
  \stateO(X) \eqdef{}& \prod_{1 \leq i \leq m} \stateO_i(X)\\
  r_{X,(i,j)}(x, (\sigma_1,\dots,\sigma_i,\dots,\sigma_m)) \eqdef{}&
\begin{cases}
  \Some(y, (\sigma_1,\dots,\sigma'_i,\dots,\sigma_m)) &\mbox{ if } r_i(x, \sigma_i) = \Some(y,\sigma')\\
  \None & \mbox{ otherwise }
\end{cases}
\end{align*}

Turning to the separation logic, we specialize the rule \ruleref{wp-reify} to the signature $G$ and the reifier $R_G$, and simplify it to
\begin{mathpar}\small
\inferH{wp-reify-local}
{\hasstate_i(\sigma_i) \and
r_i(x,\sigma_i) = \Some(y,\sigma'_i)
\and
k\ y = \Next(\beta)
\and
\later\big(\hasstate_i(\sigma'_i) \wand \wpre{\beta}{\Phi} \big)}
{\wpre{\Vis_{(i,j)}(x, k)}{\Phi}}
\end{mathpar}
where the predicate $\hasstate_i(\sigma)$ tracks the local component of the global state associated with the $i$\textsuperscript{th} reifier.
The predicates are defined to validate the following rule, which allows us to split the global state into local subcomponents and combine them back together:
$$
\hasstate(\vec\sigma) \provesIff \hasstate_1(\sigma_1)\ast \dots \ast \hasstate_m(\sigma_m).
$$

Then to write down the abstractions for the domain-specific logic in \Cref{sec:store_logic} we change the $\heapctx$ definition to link together only the state corresponding to the specific effects:
$$
\heapctx \eqdef{} \knowInv{}{\Exists \sigma. \hasstate_i(\sigma) \ast \ownGhost{}{\authfull \sigma}}
$$
where the higher-order store reifier is the $i$\textsuperscript{th} subreifier of ${\vec{R}}$.

\begin{example}
Recall the program $\mathsf{fact\_io}$ from the beginning of this section.
We use \Cref{prop:factorial_spec} 
to show the following specification:
$$
\heapctx \ast \hasstate_j(k\vec{n},\vec{m})\vdash\wpre{\mathsf{fact\_io}}{\Ret \_. \hasstate_j\left(\vec{n},(!k)\vec{m}\right)}
$$
where $\hasstate_j$ tracks the state of the input/output effects.
The specification tells us that if we run $\mathsf{fact\_io}$ with the starting state $\left(k\vec{n},\vec{m}\right)$ for the input/output effects, then we end up with the state $\left(\vec{n},(!k)\vec{m}\right)$ for the input/output effects.
\end{example}

\subsection{Modular reasoning with a generic ground type}
As we have mentioned in \Cref{sec:programming}, we often would like to work with the \gitrees $\IT_{E}(A)$ for some generic ground type $A$ that is ``large enough'' to contain ground values that we need to represent (e.g. the unit type, natural numbers, the type of locations, etc).
That is, we assume that the ground type $A$ is isomorphic to a sum $\Tunit + \Tnat + \Loc + \dots$, depending on ground values we need.
We tackle this generic ground type in a similar way we deal with different effect signatures modularly.

Specifically, we write $\subRet{B}{A}$ if $A \simeq B + C$ for some type $C$.
We then have the generic return constructor $\Rret : B \to \IT_{E}(A)$ for any $\subRet{B}{A}$.
Similarly, we have a generic ``destructor'' $\getret : \IT_{E}(A) \times (B \to \IT_{E}(A)) \to \IT_{E}(A)$ which allows us to extract a ground value of type $B$, as long as we have $\subRet{B}{A}$.
such that $\getret$ is a homomorphism in the first argument, which satisfies:
$$
\getret(\Rret(b), g) =
\begin{cases}
  g(b) \mbox{  if $b \in B$},\\
  \Err(\RunTime) \mbox{ otherwise.}
\end{cases}
$$
Then, the $\getnat$ function from \Cref{{sec:programming}} is just the specialization of $\getret$ to the situation $\subRet{\Tnat}{A}$.

The predicate  $\subRet{B}{A}$ is formalized in Coq as a typeclass, making it easy to use the generic operations like $\Rret$ and $\getret$.
In the remainder of this paper we will stick to those generalized operations, and will assume that the ground type $A$ contains all the ground values we need.

\section{Type Safety for Cross-language Interoperability}
\label{sec:interop}
One of the advantages of using \gitrees for denotational semantics is that it provides a common setting for interpreting and reasoning about different languages with different effects, and then combining the results in a modular manner.
In this section we demonstrate this point by verifying type safety of interoperability between two different languages.
The interoperability is achieved by allowing embeddings from one language into another at a particular boundary \cite{MatthewsFindler:2007}.
We take inspiration for this case study from the approach of \citet{PattersonEtAl:2022}, who consider interoperability of different languages at a level of a common third language, which both the source languages are compiled down to.
The communication between the source languages is done through \emph{glue code} at the level of the target language, which converts types from one language to another.
The type safety result then states that well-typed programs in a combined language can only go wrong due to conversion errors at the boundaries.

Specifically, we first consider an affine programming language $\afflang$ with linear references and strong updates, which we interpret in guarded interaction trees using the higher-order store effects (\Cref{ex:store_sig}).
We show the type safety of $\afflang$ by building a logical relations model.

We then consider cross-language interoperability, by allowing embedding of (higher-order) programs from the non-affine language $\iolang$ (\Cref{{sec:iolang_model}}) into $\afflang$, thus combining two languages with different type systems and different effects.
Following the approach outlined in \Cref{sec:subeffects} we \emph{reuse the standalone interpretations} of $\iolang$ and $\afflang$ to interpret the combined language in $\IT_{\mathit{store}, \mathit{io}}$.
Finally, we show type safety of this combined language, by building the logical relations model out of the models for the individual languages.

We stress that our approach here allows us to prove type safety of $\iolang$ and $\afflang$ \emph{separately}, and then prove type safety of the combined language by \emph{reusing} the logical relations for the individual languages, thus highlighting the modular nature of the \gitrees framework.

\subsection{An Affine Programming Language}
\label{sec:afflang_model}
First, we consider an affine programming language $\afflang$ with references with strong updates, and show how to interpret it in \gitrees in a way that enforces linearity.
We then consider the combination $\afflang + \iolang$, which allows us to embed $\iolang$ programs, including functions, into $\afflang$.
The syntax for the affine language $\afflang$ is as follows:
\begin{align*}
\typeAff \in \langkwa{\Type} \bnfdef{}&
\tBool \ALT \tNat \ALT \tUnit
\ALT \typeAff_1 \atensor \typeAff_2 \ALT \typeAff_1 \lolli \typeAff_2
\ALT \tRef{\typeAff}
\\
\expr \in \langkwa{\Expr} \bnfdef{}&
  n \ALT b \ALT \unittt \ALT
  \avar \ALT
\LamA \avar.\expr \ALT
\expr_1\ \expr_2 \ALT
\aPair{\expr_1}{\expr_2} \ALT
\LetA \aPair{\avar_1}{\avar_2}=\expr_1 in \expr_2 \\& \ALT \aAlloc(\expr) \ALT \aDealloc(\expr_1) \ALT
\aReplace(\expr_1,\expr_2)
\end{align*}
To differentiate between the terms of $\afflang$ and the terms of $\iolang$, we use the {\color{orange}orange} color to refer to types and programs of $\afflang$.
\begin{figure}[t]\small
  \begin{mathpar}
    \infer
    {b \in \mathbb{B}}
    {\aCtx \provesAff b : \tBool}
    \and
    \infer
    {n \in \mathbb{N}}
    {\aCtx \provesAff n : \tNat}
    \and
    \axiom
    {\aCtx \provesAff \unittt : \tUnit} 
    \and
    \axiom
    {\aCtx_1, \avar : \typeAff, \aCtx_2 \provesAff \avar : \typeAff}
    \and
    \infer
    {\avar : \typeAff_1, \aCtx \provesAff \expr : \typeAff_2}
    {\aCtx \provesAff \LamA \avar. \expr : \typeAff_1 \lolli \typeAff_2}
    \and
    \infer
    {\aCtx_1 \provesAff \expr_1 : \typeAff_1 \lolli \typeAff_2
      \and \aCtx_2 \provesAff  \expr_2 : \typeAff_1
    }
    {\aCtx_1, \aCtx_2 \provesAff \expr_1 \ \expr_2 : \typeAff_2}
    \and
    \infer
    {\aCtx_1 \provesAff \expr_1 : \typeAff_1
      \and \aCtx_2 \provesAff  \expr_2 : \typeAff_2}
    {\aCtx_1,\aCtx_2 \provesAff \aPair{\expr_1}{\expr_2} :  \typeAff_1 \atensor \typeAff_2}
    \and
    \infer
    {\aCtx_1 \provesAff \expr_1 : \typeAff_1 \atensor \typeAff_2
      \and \avar_1:\typeAff_1,\avar_2:\typeAff_2,\aCtx_2\provesAff \expr_2 : \typeAff
    }
    {\aCtx_1, \aCtx_2 \provesAff \LetA \aPair{\avar_1}{\avar_2}=\expr_1 in \expr_2 : \typeAff}
    \and
    \infer
    {\aCtx \provesAff \expr : \typeAff}
    {\aCtx \provesAff \aAlloc(\expr) : \tRef{\typeAff}}
    \and
    \infer
    {\aCtx \provesAff \expr : \tRef{\typeAff}}
    {\aCtx \provesAff \aDealloc(\expr) : \tUnit}
    \and
    \infer
    {\aCtx_1 \provesAff \expr_1 : \tRef{\typeAff_1}
      \and \aCtx_2 \provesAff \expr_2 : \typeAff_2}
    {\aCtx_1,\aCtx_2  \provesAff \aReplace(\expr_1,\expr_2) : \typeAff_1 \atensor\, \tRef{\typeAff_2}}
  \end{mathpar}
  \caption{Type system for $\afflang$.}
  \label{fig:afflang_typing}
\end{figure}
The type system for $\afflang$ is given in \Cref{{fig:afflang_typing}}.
Let us explain some of the details.
The language $\afflang$ contains Booleans, natural numbers, and the unit type.
It also features linear functions $\typeAff_1 \lolli \typeAff_2$.
Note that in the typing rule for function application, the context is split between typing of the function and typing of the argument.
This ensures that the function and its argument do not share any variables or resources in common.

The language also features linear pairs $\typeAff_1 \atensor \typeAff_2$.
In the typing rule for pairs the typing environment has to be split between the two components. 
This ensures that we cannot have, e.g., pairs of the form $\aPair{x}{x}$.

Finally, the language features references with strong updates, i.e., references that can store values of different types.
The constructors $\aAlloc$ and $\aDealloc$ are used to allocate and free the references, respectively.
To ensure linearity, we have a single operation that combines reading from a reference and performing a strong update.
The program $\aReplace(r, v)$ reads the value that is stored in the reference $r$ and updates it to the value $v$.
It then returns a linear pair consisting of the old value and the reference itself, allowing one to reuse the reference later on.

The meaning of $\afflang$ is given by the interpretation function $\Sem{\expr}_{\rho} : \IT_{F}(A)$, where $\rho$ is the environment mapping the free variables of $\expr$ to \gitrees, and where $F$ is a signature which contains the higher-order store effects (\Cref{ex:store_sig}).
We assume that the ground type $A$ contains, in addition to natural numbers and the unit type, the type $\Loc$ of locations.
The semantic interpretation of $\afflang$ is given in \Cref{fig:afflang_interp}.
In the interpretation of compound expression we split the environment $\rho$ into the environments $\rho_1$ and $\rho_2$, for the free variables of $\expr_1$ and $\expr_2$ respectively.

In order to ensure that the variables from the context $\aCtx$ are used at most once, we wrap every variable in a thunk which can be evaluated at most once:
\begin{align*}
 \Thunk(\alpha) \eqdef{}& \ALLOC\big(\Ret(0),\Lam \ell.
    \Fun(\Next(\Lam \_. \IF(\READ(\ell), \Err(Lin), \WRITE(\ell,\Rret(1)) \SEQ \alpha)))\big)\\
  \Force(\alpha) \eqdef{}& \APPs{\alpha}{\Rret(0)}
\end{align*}
When we called a thunked \gitree for the second time, it will return the error $\Err(Lin)$, signifying that we broke the linearity condition.
Here we assume that we have a separate error state $Lin\in\Error$, because we want to treat linearity condition errors separate from type errors or memory safety errors.
\begin{figure}[t]
  \begin{align*}
    \Sem{b}_\rho \eqdef{}&
                           \begin{cases}
                             \Rret(1)&\mbox{if b = $\langkwa{true}$}\\
                             \Rret(0)&\mbox{otherwise}
                           \end{cases}
    & \Sem{\unittt}_\rho \eqdef{} & \Rret(())\\
    \Sem{n} \eqdef&{} \Rret(n)
                                    &  \Sem{\LamA \avar.\expr}_\rho  \eqdef{}&
                                                                               \Fun(\Next(\Lam \alpha.\Sem{\expr}_{\rho[\avar \mapsto \alpha]})) \\
    \Sem{\avar}_\rho \eqdef{}& \Force(\rho(\avar)) 
                                    &  \Sem{\expr_1\ \expr_2}_\rho \eqdef{}&
                                                                             \begin{aligned}[t]
                                                                               &\LET x = \Sem{\expr_2}_{\rho_2} IN \\
                                                                               &\APPs{\Sem{\expr_1}_{\rho_1}}{\Thunk(x)}
                                                                             \end{aligned} \\
    \Sem{\LetA \aPair{\avar_1}{\avar_2}=\expr_1 in \expr_2}_{\rho} \eqdef{}&
                                                                             \begin{aligned}[t]
                                                                               &\LET x=\Sem{\expr_1}_{\rho_1} IN\\
                                                                               &\LET y=\Thunk(\Proj{1}(x)) IN\\
                                                                               &\LET z=\Thunk(\Proj{2}(x)) IN\\
                                                                               &\Sem{\expr_2}_{\rho_2[\avar_1\mapsto y,\avar_2\mapsto z]}
                                                                             \end{aligned}
    &  \Sem{\aPair{\expr_1}{\expr_2}}_{\rho} \eqdef{}&
                                                       (\Sem{\expr_1}_{\rho_2},\Sem{\expr_2}_{\rho_2}) \\
    \Sem{\aAlloc(\expr)}_{\rho} \eqdef{}&
                                          \LET x=\Sem{\expr}_\rho IN \ALLOC(x,\Rret)
                                    &  \Sem{\aDealloc(\expr)}_\rho \eqdef{}&
                                                                             \begin{aligned}[t]
                                                                               &\getret(\Sem{\expr}_\rho,\DEALLOC)
                                                                             \end{aligned}\\
    \Sem{\aReplace(\expr_1,\expr_2)}_{\rho} \eqdef{}& \multispan{3}{$\begin{aligned}[t] &\LET y = \Sem{\expr_2}_{\rho_2} IN \getret(\Sem{\expr_1}_{\rho_1}, \Lam \ell.\LET x = \READ(\ell) IN\\ & \quad\WRITE(\ell,y) \SEQ (x, \Rret(n)))\end{aligned}$}
  \end{align*}
  
  \caption{Interpretation of $\afflang$.}
  \label{fig:afflang_interp}
\end{figure}
As such, in the interpretation of a function application we put the argument in a $\Thunk$, 
and whenever we use the argument (or any affine variable) we then have to $\Force$ it.

We can show that if we have a well-typed program, then it does not have any run-time errors, and that all the thunks are evaluated at most once:
\begin{proposition}
  \label{prop:afflang_safety}
  Suppose that ${}\provesAff e : \typeAff$, and suppose that
$(\sigma,\Sem{e}) \istep^\ast (\sigma',\beta)$.
Then $\beta \neq \Err(err)$.
\end{proposition}
To prove \Cref{prop:afflang_safety} we use a logical relation, given in \Cref{fig:afflang_logrel}, defined similarly to the logical relation from \Cref{sec:adequacy}.
The interpretation $\AVV{-}$ of the base types cover the appropriate subsets of natural numbers.
Reference types are interpreted using the ``pointsto'' $\ell \mapsto  \alpha_v$ proposition, and affine pairs are interpreted component-wise.
The main differences to note here are:
\begin{enumerate*}[label=(\arabic*)]
\item variables in $\afflang$ are interpreted as thunks, and thus we adjust the interpretation $\AVVCtx{\aCtx}$ to account for that;
\item values in $\afflang$ can be used at most once; hence
 the value interpretation is not persistent, i.e., there is no persistently modality $\Box$ in the interpretation of function types.
\end{enumerate*}
\begin{figure}[t]\small
  \begin{align*}
    \AVV{\tUnit}(\beta_v) \eqdef{}& \beta_v = \Rret(())
&  \AVV{\tNat}(\beta_v) \eqdef{}& \Exists n \in \mathbb{N}. \beta_v = \Rret(n)
\\  
  \AVV{\type_1 \lolli \type_2}(\beta_v) \eqdef{}& 
                                                     \All \alpha_w. \AVV{\type_1}(\alpha_w) 
                                                     \wand \AEE{\type_2}(\APPs{\beta_v}{\alpha_w})
& \AVV{\tBool}(\beta_v) \eqdef{}& \beta_v = \Rret(0) \vee \beta_v = \Rret(1)
    \\ \AVV{\tRef{\type}}(\beta_v) \eqdef{}&
      \begin{aligned}[t]
      &\Exists \ell \in \Loc,\, \alpha_v. \big(\beta_v = \Rret(\ell))\big) \ast{}\\
      & \ell\mapsto \alpha_v \ast \AVV{\type}(\alpha_v)
      \end{aligned}  
& \AVV{\type_1 \atensor \type_2}(\beta_v) \eqdef{}&
  \begin{aligned}[t]
    &\Exists \gamma_v,\delta_v. \beta_v = (\gamma_v,\delta_v) \ast{}\\
    &\AVV{\type_1}(\gamma_v) \ast \AVV{\type_2}(\delta_v)
  \end{aligned}
\\\AEE{\type}(\alpha) \eqdef{}& 
      \heapctx \wand \wpre{\alpha}{\Ret \beta_v. \AVV{\type}(\beta_v)}
& \protec(\Phi)(\beta_v) \eqdef{}& \wpre{\Force(\beta_v)}{\Phi} 
\\ \AVVCtx{\aCtx}(\rho) \eqdef{}&\All (\avar :\type) \in \aCtx.
   \protec(\AVV{\type})(\rho(\avar))
&  \aCtx \modelsAff \alpha: \typeAff \eqdef{}
      \All \rho.{}& \AVVCtx{\aCtx}(\rho) \implies
          \AEE{\type}(\alpha(\rho))
  \end{align*}  
\caption{Logical relation for $\afflang$.}
\label{fig:afflang_logrel}
\end{figure}
\begin{lemma}[Fundamental property]
\label{lem:fundamental_aff}
  For any expression $e$, if $\aCtx \provesAff e : \typeAff$, then $\aCtx \modelsAff \Lam \rho.\Sem{\alpha}_\rho : \typeAff$.
\end{lemma}
We prove the fundamental property by induction on the typing derivation.
More specifically, for each typing rule we prove an associated \emph{compatibility lemma}, by replacing expressions with interaction trees and the derivability $\provesAff$ with validity $\modelsAff$.
For example, the compatibility lemma for $\aDealloc$ looks as follows:
\begin{lemma}
  Suppose that $\aCtx \modelsAff \alpha : \tRef{\typeAff}$.
Then $\aCtx \modelsAff \Lam \rho. \getret(\alpha(\rho),\DEALLOC) : \tUnit$.
\end{lemma}
Proving all the compatibility lemmas is relatively straightforward using the separation logic rules.
Having separate compatibility lemmas will be useful for us in the next section.

By combining the fundamental property with the safety theorem for the weakest precondition calculus (\Cref{thm:wp_adequacy}) we obtain a proof of \Cref{prop:afflang_safety}.

\paragraph{Safety for $\iolang$.}
Similar to the logical relation for safety for $\afflang$, we define a logical relation for $\iolang$.
It is simply an unary version of the logical relation from \Cref{sec:adequacy}.
For example, the expression relation is defined as
$$
\EE{\type'}(\alpha) \eqdef \All \sigma'.\hasstate_i(\sigma') \wand \wpre{\alpha}{\Ret \beta_v. \Exists \sigma'_1. \VV{\type'}(\beta_v) \ast \hasstate_i(\sigma'_1)}
$$
where $\hasstate_i$ tracks the state for the input/output effects.
We omit the other details here and direct an interested reader to the Coq formalization.
We only note that just like for $\afflang$, we split the proof of the fundamental property into compatibility lemmas, which we will use in the next section.

\subsection{Interoperability Between Languages}
Following the approach of \citep{PattersonEtAl:2022}, we allow the interaction between the languages $\iolang$ and $\afflang$, using the guarded interaction trees as the ``common ground'', combining the effects of the two languages.
At the syntactic language level, this is done by allowing one to embed the expressions of $\iolang$ into the $\afflang$ programs.
The embedding is given by the following typing rule:\footnote{For simplicity, we only consider one-way embeddings from $\iolang$ to $\afflang$, and we only consider embeddings of closed terms. This simplifies the type system, but does not lead to loss of expressivity, since we allow type conversions of functions.}
\begin{mathpar}
\inferH{typed-conv}
{\proves e : \tau' \and \tyconv{\tau'}{\tau}}
{\aCtx \provesAff \aEmbed{e}_{\tyconv{\tau'}{\tau}} : \langkwa{\tau}}
\end{mathpar}
The crucial ingredient for the interaction is a type conversion relation $\tyconv{\tau'}{\tau}$
stating that the $\iolang$ type $\tau'$ is convertible to the $\afflang$ type $\langkwa{\tau}$.

We have the following type conversions:
\begin{mathpar}
\axiom{\tyconv{\Tnat}{\tNat}}
\and
\axiom{\tyconv{\Tnat}{\tUnit}}
\and
\axiom{\tyconv{\Tnat}{\tBool}}
\and
\infer
{\tyconv{\type'_1}{\type_1} \and \tyconv{\type'_2}{\type_2}}
{\tyconv{(\Tnat \to \type'_1) \to \type'_2}{\type_1 \lolli \type_2}}
\end{mathpar}
The first three type conversions say that we can freely convert between integers, Booleans, and the unit type (since all of them have similar internal representation).
The last conversion is more interesting.
It says that we can convert between affine functions and non-affine functions.
The affine argument $\langkwa{\type_1}$ is represented as a thunk $(\Tnat \to \type'_1)$, which we will protect at runtime to ensure that it is not invoked more than once.
The type $\Tnat$ is used in lieu of the unit type which is absent from $\iolang$.

\paragraph{Glue code for conversion functions}
In order to
\begin{enumerate*}
\item convert between different types, and 
\item ensure the linearity of $\afflang$ programs that might cross the boundary to $\iolang$,
\end{enumerate*}
we need to interpret embedded terms with additional \emph{glue code}.
For every type conversion $\tyconv{\type'}{\type}$ we generate recursively a pair of conversion functions $\toAff_{\type',\langkwa{\type}}$ and $\fromAff_{\langkwa{\type},\type'}$ which convert the representations from $\type'$ to $\langkwa{\type}$ and vice versa.
The glue code for converting between the base types ensures that the underlying natural number representation stays within the range.
For example, for $\tyconv{\Tnat}{\Tbool}$ we have:
\begin{align*}
 \toAff_{\Tnat,\tBool}(\alpha) \eqdef{}& \IF(\alpha,\Rret(1),\Rret(0))
& \fromAff_{\tBool,\Tnat}(\alpha) \eqdef{}& \alpha
\end{align*}
The glue code for converting functions is a bit more involved:\footnote{The glue code for functions is a bit more complicated than in \cite{PattersonEtAl:2022}. Specifically, they do not protect the converted affine functions. This is fine in their setting, because their affine functions are pure, and invoking them multiple times does not lead to run-time errors. However, in the presence of references with strong updates this assumption is no longer true, and not protecting the converted functions will lead to unsound behavior!}
\begin{align*}
  \toAff_{{(\Tnat \to \type'_1) \to \type'_2},\langkwa{\type_1 \lolli \type_2}}(\alpha)
\eqdef{}&
\LET p = \alpha IN \\ \Fun(&\Next(\Lam x.\LET y = \fromAff_{\langkwa{\type_1},\type'_1}(\Force(x)) IN 
 \toAff_{\type'_2,\langkwa{\type_2}}(\APPs{p}{\Thunk(y)})))\\
\fromAff_{\langkwa{\type_1 \lolli \type_2},{(\Tnat \to \type'_1) \to \type'_2}}(\alpha)
\eqdef{}&
 \LET p = \alpha IN \LET p' = \Thunk(p) IN \\
\Fun(\Next(&\Lam x. \LET f = \Force(p) IN \LET y = \toAff_{\type'_1,\langkwa{\type_1}}(\Force(x)) IN \fromAff_{\type'_2,\langkwa{\type_2}}(\APPs{f}{\Thunk(y)})))
\end{align*}
When we convert a function, we need to recursively convert the argument and the result; in addition the argument to the function needs to be protected with a $\Thunk$.
Furthermore, when affine functions are converted to non-affine ones, we need to protect the function itself with a $\Thunk$, to ensure that the function is not invoked multiple times.
Calling an affine function multiple times might be unsound, e.g., calling the following function twice will attempt to deallocate an already deallocated reference:
$$
(\LamA \ell. \LamA \_.(\LamA \_. 7)\ \aDealloc(\ell))\ (\aAlloc(42)) : \tUnit \lolli \tNat.
$$

\paragraph{Partial safety for the combined language}
We interpret the combined language with embeddings using the glue code.
The embedding from $\iolang$ to $\afflang$ is interpreted as follows:
$$
\Sem{\aEmbed{e}_{\tyconv{\type'}{\type}}}_{\rho} =
  \toAff_{\type',{\langkwa{\type}}}(\IOSem{e}),
$$
where $\IOSem{-}$ is the interpretation function for $\iolang$ expressions from \Cref{{sec:iolang_model}}.
The interpretation for all the operations, except for the embedding, remains unchanged.

Using the modular approach we described in \Cref{sec:subeffects}, we interpret the extended language $\iolang+\afflang$ into the guarded interaction trees $\IT_{\mathit{store}, \mathit{io}}$
with the signature that combines the input/output effects and the store effects.
This ensures that the interpretation of the languages end up in the same domain, where they can interact.
By combining the reifiers for the effects of the individual languages we also get the reduction relation
 $((\sigma_1,\sigma'_1),\alpha) \istep ((\sigma_2,\sigma'_2), \beta)$, where $\sigma_1$ and $\sigma_2$ are stores, and $\sigma'_1$ and $\sigma'_2$ are input/output tapes.

 Of course, in the combined language with conversions, our programs can actually violate the linearity condition, since it is no longer enforced by the type system.
 However, we can prove that linearity violations at the boundary are the only errors that we will possibly get.
Thus we will to prove the following safety theorem:
\begin{theorem}
\label{prop:combined_safety}
Suppose that $\ \provesAff e : \typeAff$ with the embedding rule, and suppose that
$((\sigma_1,\sigma'_1),\Sem{e}) \istep^\ast ((\sigma_2,\sigma'_2),\beta)$.
Then either $\beta$ is not an error, or $\beta = \Err(Lin)$.
\end{theorem}

\subsection{Logical Relations for the Combined Safety}
\label{sec:comb_safety}
We will prove \Cref{prop:combined_safety} by constructing a logical relation similarly to what we did for the individual languages in \Cref{sec:afflang_model}.
Our goal is to do so modularly, by reusing as much material from \Cref{sec:afflang_model} as possible. 
In particular, we will reuse all the compatibility lemmas we used to prove \Cref{lem:fundamental_aff}, and only prove one (!) new compatibility lemma for \ruleref{typed-conv}.
Letting $\EE{\type'}$ and $\AEE{\type}$ be the expression relations for the logical relation for $\iolang$ and $\afflang$ resp., this compatibility lemma is:
\begin{lemma}
\label{lem:compat_tyconv}
Suppose that $\EE{\type'}(\alpha)$ and $\tyconv{\type'}{\type}$.
Then $\AEE{\type}(\toAff_{\type',\langkwa{\type}}(\alpha))$.
Moreover, for the other direction, suppose that $\AEE{\type}(\alpha)$ and $\tyconv{\type'}{\type}$.
Then $\EE{\type'}(\fromAff_{\langkwa{\type},\type'}(\alpha))$.
\end{lemma}

When we presented the separation logic and logical relations earlier, we presented a slightly simplified version which was sufficient for our purposes.
However, in order to prove \Cref{lem:compat_tyconv} we need to make use of features that we have not yet presented.
We describe those features now.

\paragraph{Separation logic for weak safety}
The first feature is that our notion of weakest precondition is actually parameterized by
a \emph{stuckness predicate} $s: \Error \to \Prop$, and satisfies the following rule
(in addition to the rules presented in  \Cref{sec:program_logic}):
\[
\axiom
{s(e) \vdash \wpre{\Err(e)}[s]{\Phi}}
\]
This means that if the stuckness predicate $s$ holds for some error $e$, then the weakest precondition predicate holds for that error, irrespectively of the postcondition.
The earlier presented weakest precondition, which did not allow for errors, is obtained by using
$s(e) = \FALSE$.
 The general weakest precondition with the stuckness predicate satisfies a version of the adequacy/safety property (\Cref{thm:wp_adequacy}), in which the \emph{(safety)} condition is replaced with the following condition:
\begin{itemize}
\item \emph{(weak safety)} 
either there are $\beta_1$ and $\sigma_1$ such that
    $(\beta,\sigma')\istep(\beta_1,\sigma_1)$,
or $\beta = \Err(e)$ with $e$ satisfying the predicate $s$.
\end{itemize}

All the logical relations presented earlier are actually parameterized by a predicate $s$
and uses $\wpre{\alpha}[s]{\Phi}$ in the expression interpretation --- to recover
the earlier stated theorems for full safety we simply instantiate the logical relations with $s = \Lam e. \FALSE$.

\paragraph{Freely adjoined Kripke structure}
In addition to the stuckness bit, we actually formulate our logical relations a bit more generally than what we have shown so far.
This is because in the logical relations for individual languages require particular resources that we need to combine when constructing a logical relation for the combined language.

Indeed, to ensure that our logical relations are sufficiently modular, we parameterize the expression relation by an arbitrary predicate $P : A \to \Prop$ of an arbitrary type $A$.
We refer to this predicate $P$ as freely adjoined Kripke structure (because, in the underlying model of Iris, it allows us to make arbitrary transitions between worlds, constrained by the predicate $P$).
The general definition of the expression relation for all our logical relation is thus:
$$
\EK{s}{P}{\Phi}(\alpha) \eqdef
\All x:A. P(x) \wand \wpre{\alpha}[s]{\Ret \beta_v. \Exists y:A. \Phi(\beta_v) \ast P(y)}.
$$
The idea is that the $P$ parameter describes additional resources (for other effects) and ensures
that ITrees in the expression relation preserve any such additional resources.
The idea is akin to the ``baking-in'' of the frame rule in models of separation logic for higher-order languages \cite{DBLP:journals/lmcs/BirkedalY08,DBLP:conf/icalp/BirkedalRSY08}. 

The expression relations for individual languages $\iolang$ and $\afflang$ are then both parameterized by
predicates $P : A \to \Prop$ and $s: \Error \to \Prop$ and defined as
\begin{align*}
  \EE{\type'}(\alpha) \eqdef{}& \EK{s}{\Lam (\sigma',x). \hasstate_i(\sigma') \ast P(x)}{\VV{\type'}}(\alpha)\\
  \AEE{\type}(\alpha) \eqdef{}& \EK{s}{\Lam x. \heapctx \ast P(x)}{\AVV{\type}}(\alpha).
\end{align*}

For $\iolang$ and for $\afflang$ we prove the compatibility properties for arbitrary $P$ and $s$
(in the proofs of the compatibility lemmas, the resources described by $P$ are just passed through).
We obtain full safety for the individual languages by instantiating $P$ with $\Lam x.\TRUE$ and $s$ with $\Lam e.\FALSE$.

When we prove partial safety for the combined language, we reuse the same logical relations and the same compatibility lemmas for the individual languages,
by instantiating $P$ in $\EE{\type'}$ with $P(x) =\heapctx$, and by instantiating $P$ in $\AEE{\type}$ with $P(x) = \hasstate_i(x)$, and $s$ with $\Lam e. (e = Lin)$ in both cases.
Then, to prove the fundamental property of the combined language we can reuse the compatibility lemmas for individual languages, and it only remains
to prove the compatibility \Cref{lem:compat_tyconv} for the type conversion.

The most interesting case of \Cref{lem:compat_tyconv} is the conversion between functions, which involves showing  soundness of the glue code.
The interpretation of non-affine functions is persistent, as it begins with $\Box$, since it is expected that you can use non-affine functions multiple times.
The interpretation of affine functions, however, is not persistent ---  functions can be invoked only once.
Because of that, we cannot directly use the interpretation $\AVV{\type_1 \lolli \type_2}$ when proving $\VV{(\Tnat \to \type'_1)  \to \type'_2}$.
Instead, we put the interpretation of the affine function in a persistent invariant, which states:
$$
(\ell \mapsto \Rret(0) \ast \AVV{\type_1 \lolli \type_2}(\alpha_v)) \vee (\ell \mapsto \Rret(1)).
$$
This describes when the affine function $\alpha_v$ is protected with the $\Thunk$ via a reference $\ell$:
either the function has not been invoked yet ($\ell \mapsto \Rret(0)$) and it satisfies the value interpretation, or the function has already been invoked ($\ell \mapsto \Rret(1)$) and forcing its thunk will result in $\Err(Lin)$, in which case we do not invoke the function.
(As a side remark, we note that we here crucially rely on Iris's powerful notion of invariants --- this is another example of why it is advantageous to use Iris as the basis for our separation logic.)

Having established the compatibility lemma for type conversions, and the compatibility lemmas for the operations for each individual language, we prove the fundamental property for the logical relation for the combined language.
In particular, for a closed term ${}\provesAff e : \typeAff$ we obtain
$\EK{\Lam e. e = Lin}{\Lam x. \heapctx \ast \hasstate_i(x)}{\AVV{\type}}(\Sem{e})$.
From the adequacy property of the weakest precondition 
we can then conclude \Cref{prop:combined_safety}.

In summary, this approach to logical relations, with freely adjoined Kripke structure as a parameter, allows us to scale proofs to interoperability between multiple different languages in a modular way.
For each individual language we can prove safety separately (without knowing in advance with which other languages we are going to interface).
Then we can combine logical relations for individual languages together, by instantiating the freely adjoined Kripke structure with the shared resources or with resources needed to verify the glue code, and reusing the compatibility lemmas.

\section{Discussion and Related Work}
\label{sec:related_work}
We have already discussed a lot of related work throughout the paper; in this section we include some further discussion
of related work.

\paragraph{Differences with interaction trees.}
Whilst our work takes direct inspiration from the interaction trees approach, there are some crucial differences.
One of the main difference comes from the treatment of effect reification.
The original type of interaction trees is a monad, and the effects in an interaction tree can be reified ``in one go'', for example, with a state monad transformer over ITrees.
In our case, we cannot reify all the effects in a guarded interaction tree, due to higher-order functions and higher-order effects.
For example, a guarded interaction tree can be a function that contains latent effects, but these effects can only be reified at the point when the function is invoked.
Because regular interaction trees contain only first-order structures, it is possible to traverse them completely, reifying all the effects.

Another difference worth mentioning, is that regular interaction trees can be extracted and executed from Coq.
Our formalization cannot be directly extracted, as it is built upon a layer of guarded type theory.
One potential approach would be to erase the guarded parts from the formulation of \gitrees and obtain that way a representation of \gitrees in a functional language like OCaml or Haskell, which already supports mixed-variance datatypes.
Then, the extraction can be set up in such a way as to use this representation.
We have not researched this direction and leave it to future work.

Finally, an important difference between our work and that on ITrees, is that the ITrees development relies heavily on the weak bisimilarity theory of interaction trees, while we opt for developing separation logic and refinements instead.
There are several things that complicate the study of bisimilarity for guarded interaction trees.
Firstly, the higher-order nature of \gitrees suggests that we need to study a more involved notion of behavioral equivalence, like applicative bisimilarity.
Secondly, we believe that developing a theory of \emph{weak} bisimilarity in the context of Iris and guarded type theory is still an open question, complicated by issues with step-indices.
For these reasons we believe that developing weak applicative bisimularity for guarded interaction trees will require new techniques and we leave it for future work.

\paragraph{Differences with the standard Iris approach to verification.}
The standard Iris-based approach to separation logic \cite{DBLP:conf/popl/JungSSSTBD15,DBLP:conf/icfp/0002KBD16,DBLP:conf/esop/Krebbers0BJDB17} is based on operational semantics, and has been proven to scale well to complicated programming language features.
The main difference with our work, is that we are the first to build an Iris-based separation logic over denotational semantics, in a way that is tightly integrated with the existing Iris ecosystem.
In particular, we rely on the Iris ecosystem for various data structures, resource algebras, base logic (but not the program logic), and the Iris Proof Mode.
As such, in terms of reasoning about specific concrete programs, the \gitrees-based approach is not that different from what a normal Iris user expects, with the added advantage of using equational reasoning for many computation steps that usually requires
some form of symbolic execution.

The main advantage of our approach comes into play when we want to
consider new models of programming languages, or reasoning about
programs with combinations of effects (as in \Cref{{sec:subeffects}}), or
reasoning about interoperability (as in \Cref{{sec:interop}}).

\paragraph{Domain theory and guarded type theory.}
Guarded type theory \cite{DBLP:journals/corr/BizjakGCMB16,DBLP:conf/lics/BirkedalM13} has been studied as a setting for domain theory before \cite{DBLP:journals/corr/abs-1208-3596,PaviottiEtAl:2015,MogelbergPaviotti:2019,MogelbergVezzosi:2021}, but previous works mostly focused on specific typed models and was not formalized.
In contrast, here we work with guarded interaction trees as a ``universal domain'', similar to domain theoretic models of untyped $\lambda$-calculus.

The previous work used (dependent) guarded type theory not just for modeling, but also for reasoning about programs.
This required a more complicated type theory and precluded the work from being formalized in a traditional proof assistant like Coq.
By contrast, our reasoning is done in the logic \emph{over} a guarded type theory.
This is arguably simpler, and allowed us to make use of Iris and formalize all of our results in Coq.

\paragraph{Logical relations and language interoperability.}
Our case study on language interoperability in \Cref{sec:interop} is inspired by the seminal work of \citet{PattersonEtAl:2022}.
We believe that the approach we take in our work is more modular.
Firstly, our approach here is less syntactic, as we use the domain theory of guarded interaction trees as the common setting for
interpreting different languages with different effects, and we do not need to come up with a target language for each pair of source languages for which we wish to set up interoperability.
Secondly, the models that we construct for individual languages are ``local'', and that is exactly what allows for true reuse of proofs and for constructing a common model for the combined language from the individual models.
This opens up the possibility of a model for a single language to be reused for different cross-language interactions.
In contrast, the type safety of individual source languages in \cite{PattersonEtAl:2022} requires to have a model for the combined language in advance.
And finally, the treatment of effects and step-indices is more abstract and high-level in our work, since we construct our models using separation logic.

\paragraph{(Guarded) Interaction Trees and Algebraic Effects and Handlers.}
The treatment of effects in interaction trees is reminiscent of effects in programming languages based on algebraic effects and handlers \cite{PlotkinPretnar:2013,BauerPretnar:2015}.
Algebraic effects and handlers have been extensively studied in various contexts, including separation logic \cite{deVilhenaPottier:2021}, and both higher-order effects \cite{WuEtAl:2014,vandenBergEtAl:2021,BachPoulsenvanderRest:2023} and reasoning about combinations of effects \cite{DBLP:conf/lics/JohannSV10,YangWu:2021} have been investigated.
Despite the aesthetic and moral similarities to (guarded) interaction trees, there are some substantial differences between the two approaches.
Under algebraic effects and handlers, both the representation and reification of effects is done \emph{inside} the programming language.
As such, a particular theory and implementation of algebraic effects is always tied to a specific programming language.
Whereas in the interaction trees-based approach, the effects are handled in the ambient type theory, outside the type of the (guarded) interaction trees itself.
See also the discussion in \cite[Section 8.2]{XiaEtAl:2019}.

To our knowledge, the two approaches have not been formally compared.
It would be interesting to consider a denotational model of a programming language with algebraic effects inside guarded interaction trees, and to see what kind of properties can be proved using such a model.

Additionally, such a comparison might help us understand the exact class of effects that can be represented with the \gitrees-based approach.
As it currently stands, our approach to representing effects is ``open-ended'', in the sense that we can consider different classes of effects by varying the reification procedure.
Of course, different classes of effects allow for different reasoning principles.
For example, as we mentioned in the end of \Cref{sec:reductions}, we consider context-independent effects, which preclude us from modeling call/cc, but allows us to use the bind rule for the weakest precondition calculus.
We leave a formal comparison with algebraic effects and further investigations in that area to future work.

\begin{acks}                            This work was supported in part by a Villum Investigator grant (no.
  25804), Center for Basic Research in Program Verification (CPV), from the VILLUM Foundation.
  The authors are grateful to the anonymous reviewers for their comments and suggestions.
\end{acks}

\section*{Data Availability Statement}
The Coq formalization corresponding to this article is available as a Git repository at \url{https://github.com/logsem/gitrees/tree/popl24} (tag \texttt{popl24}), or under a permanent DOI at \url{https://doi.org/10.5281/zenodo.10124427}.

\end{document}